\pdfoutput=1
\documentclass[acmsmall,screen,nonacm]{acmart-primary/acmart} 

\usepackage{amsmath,amstext}

\usepackage{amssymb}
\usepackage{proof}
\usepackage{wrapfig}
\usepackage{framed}
\usepackage{tikz}
\usetikzlibrary{arrows,automata,backgrounds,calc,chains,fadings,folding,decorations.fractals,fit,patterns,positioning,mindmap,shadows,shapes.geometric,shapes.symbols,through,trees,decorations.markings,decorations.text}
\colorlet{dblue}{blue!40!black}

\usetikzlibrary{decorations.pathmorphing}
\usepackage{float}
\usepackage{stackrel}
\usepackage{enumitem}

\colorlet{lightblue}{blue!50!white}

\newcommand{\size}[1]{\#(#1)}
\newcommand{\trap}[1]{\mathrm{trap}(#1)}
\newcommand{\cfc}[1]{\mathrm{cfc}(#1)}
\newcommand{\srel}{\rightsquigarrow}
\newcommand{\DCR}{\text{DCR}}
\newcommand{\cf}[1]{\mathrm{cf}(#1)}
\newcommand{\substep}{\rightsquigarrow}
\newcommand{\lsubstep}{\leftsquigarrow}
\providecommand{\leftsquigarrow}{\mathrel{\reflectbox{$\rightsquigarrow$}}}
\newcommand{\junction}[3]{\xi_{#1}(#2,#3)}

\newcommand{\hub}[1]{\mathsf{hub}(#1)}
\newcommand{\branch}[2]{\mathsf{branch}(#1,#2)}
\newcommand{\trees}[1]{\mathsf{trees}(#1)}
\newcommand{\leaf}[1]{\mathsf{leaf}(#1)}
\newcommand{\node}[2]{\mathsf{node}(\set{#1,#2})}
\newcommand{\sjoins}{\mathsf{joins}}
\newcommand{\joins}[1]{\sjoins(#1)}
\newcommand{\pmap}{\rightharpoonup}
\newcommand{\adet}{\hookrightarrow}
\newcommand{\ladet}{\hookleftarrow}

\newcommand{\nat}{\mathbb{N}}

\newcommand{\funap}[2]{#1(#2)}

\colorlet{corange}{orange!70!red}
\colorlet{cred}{red}
\colorlet{cdarkred}{red!70!black}
\colorlet{clightred}{red!70}
\colorlet{cpurple}{blue!50}
\definecolor{cblue}{rgb}{0,0.4,0.7}
\definecolor{clightblue}{rgb}{0,0.6,1.0}
\colorlet{cgreen}{green!60!black}
\colorlet{cdarkgreen}{green!60!black}

\newcommand{\step}{\rightarrow}
\newcommand{\lstep}{\leftarrow}
\newcommand{\steps}{\mathrel{\wash{\rightarrow}\rightarrow}}
\newcommand{\lsteps}{\mathrel{\wash{\leftarrow}\leftarrow}}
\newcommand{\cstep}{\leftrightarrow}

\catcode`\@=11
\newcommand{\wash}{\mathpalette\mywash}
\newcommand{\mywash}[2]{\setbox0=\hbox{$\m@th#1{#2}$}\wd0=0pt\box0}
\catcode`\@=12

\newcommand{\sdom}{\mit{dom}}
\newcommand{\dom}{\funap{\sdom}}

\newcommand{\set}[1]{\{\, #1 \,\}}
\newcommand{\pair}[2]{\langle #1,#2 \rangle}

\newcommand{\family}[1]{\{ #1 \}}

\newcommand{\prop}{\mathrm}

\newcommand{\CR}{\ensuremath{\prop{CR}}}

\newcommand{\height}[2]{\lfloor #2 \rfloor_{#1}}

\tikzset{ao/.style={at=(#1.north),anchor=south,inner sep=0.5mm}}
\tikzset{bo/.style={at=(#1.south),anchor=north,inner sep=0.5mm}}
\tikzset{ro/.style={at=(#1.east),anchor=west,inner sep=0.5mm}}
\tikzset{lo/.style={at=(#1.west),anchor=east,inner sep=0.5mm}}

\tikzset{alo/.style={at=(#1.north west),anchor=south east,inner sep=0.5mm}}
\tikzset{blo/.style={at=(#1.south west),anchor=north east,inner sep=0.5mm}}
\tikzset{aro/.style={at=(#1.north east),anchor=south west,inner sep=0.5mm}}
\tikzset{bro/.style={at=(#1.south east),anchor=north west,inner sep=0.5mm}}

\tikzset{tloop/.style={out=60,in=120,looseness=5}}
\tikzset{bloop/.style={out=-60,in=-120,looseness=5}}
\tikzset{lloop/.style={out=210,in=150,looseness=5}}
\tikzset{rloop/.style={out=-30,in=30,looseness=5}}

\input{tikz_settings.xet}
\newcommand{\mit}{\mathit}

\begin{document}

\theoremstyle{acmplain}
\newtheorem{open}[theorem]{Open Problem}
\newtheorem{question}[theorem]{Question}
\theoremstyle{acmdefinition}
\newtheorem{remark}[theorem]{Remark}
\newtheorem{notation}[theorem]{Notation}
\newtheorem{convention}[theorem]{Convention}

\title{Decreasing Diagrams are Complete for Confluence}

\author{J\"org Endrullis}
\affiliation{%
  \institution{Vrije Universiteit Amsterdam}
  \city{Amsterdam}
  \country{The Netherlands}}
\email{j.endrullis@vu.nl}

\author{Ievgen Ivanov}
\affiliation{%
  \institution{Taras Shevchenko National University of Kyiv}
  \city{Kyiv}
  \country{Ukraine}}
\email{ivanov.eugen@gmail.com}

\author{Femke van Raamsdonk}
\affiliation{%
  \institution{Vrije Universiteit Amsterdam}
  \city{Amsterdam}
  \country{The Netherlands}}
\email{f.van.raamsdonk@vu.nl}

\renewcommand{\shortauthors}{Endrullis, Ivanov, and van Raamsdonk}

\begin{abstract}
  \emph{Confluence} is a fundamental property of nondeterministic computations,
arising from parallelism, concurrency, or freedom in the evaluation order.
It guarantees that such a computation always yields the same result, regardless of the order in which steps are taken.
The \emph{decreasing diagrams} technique of van Oostrom is one of the most versatile methods 
for establishing confluence of \emph{transition systems} (\emph{abstract rewriting systems}).
It reduces global confluence to local confluence: 
a system is confluent whenever its transitions admit a \emph{locally decreasing labeling}.  
Essentially all classical confluence criteria arise as corollaries.

A central question, posed by van~Oostrom in 1993, asks whether the decreasing diagrams technique is \emph{complete}: 
Does every confluent transition system admit a locally decreasing labeling?
This is Problem~56 of the \emph{RTA List of Open Problems}.
A positive answer was known only for countable systems, and recently up to the first uncountable cardinal~$\aleph_1$. 
The general case remained open.

We settle this thirty-three-year-old problem in full.
We prove that \emph{every} confluent transition system admits a locally decreasing labeling using only \emph{three} labels.
This bound is optimal, as two labels do not suffice even at the first uncountable cardinality.
It follows that this single criterion can, in principle, certify the confluence of \emph{every} confluent system, and hence of every confluent program.
The entire development is machine-checked in the Isabelle/HOL and Lean proof assistants
and relies only on classical logic and the axiom of choice.

\end{abstract}



\maketitle

\begin{anonsuppress}
\begin{center}
  \emph{Dedicated to the memory of Jan Willem Klop (1945--2025),
  dear friend and colleague, and a founder of the theory of term rewriting,
  for whom the completeness of decreasing diagrams was a favourite open
  problem.}
\end{center}
\end{anonsuppress}

\section{Introduction}\label{sec:intro}
\emph{Confluence}, also known as the \emph{Church--Rosser} property, is of central importance to the theory of computation.
A transition system~$\pair{S}{{\to}}$, where ${\to}$ is a binary relation on a set of states $S$, is \emph{confluent} if
two finite computations from a common state can always be brought back together:
${\lsteps}\cdot{\steps} \subseteq {\steps}\cdot{\lsteps}$, or equivalently
\begin{center}
  \begin{tikzpicture}[default,thick,baseline=0ex,every node/.style={rectangle},inner sep=0.5mm,node distance=15mm]
    \node (a) {$a$};
    \node (b) [below left of=a] {$b$};
    \node (c) [below right of=a] {$c$};
    \node (d) [below left of=c] {$d$};

    \draw [->>] (a) to (b);
    \draw [->>] (a) to (c);
    \draw [->>,densely dotted] (b) to (d);
    \draw [->>,densely dotted] (c) to (d);
  \end{tikzpicture}
\end{center}
A confluent nondeterministic system is deterministic \emph{in its results}:
despite any possible parallelism, concurrency, freedom in the evaluation order, or nondeterministic choices in such a system, all terminating computations that begin in the same state yield the same outcome.

Confluence is load-bearing in the meta-theory behind programming languages and verification.
The Church--Rosser theorem for the $\lambda$-calculus underpins the determinacy of the evaluation of functional programs,
and the consistency of equational reasoning about them~\cite{church:rosser:1936,bare:1984}, 
justifies lazy and optimal evaluation strategies~\cite{maranget:1991,ariola:etal:1995}, 
and makes type checking decidable in dependently typed languages~\cite{sozeau:etal:2020}.
Beyond the $\lambda$-calculus, confluence separates genuine nondeterminism from mere interleaving~\cite{meseguer:1992,rosu:serbanuta:2010}, 
governs the convergence of replicated data~\cite{shapiro:etal:2011,burckhardt:etal:2014} 
and state-space reduction in model checking~\cite{blom:pol:2002}, 
and is a correctness criterion for constraint solvers~\cite{abdennadher:etal:1999} and
Knuth--Bendix completion~\cite{knuth:bendix:1970}.

Concurrent and distributed programs are a prime source of confluence questions, 
since a program induces a transition system on its states whose confluence can be subtle. 
As a small example, take two processes computing greatest common divisors over shared variables:
\begin{align*}
  \mathrm{GCD}(A,B) \parallel \mathrm{GCD}(B,C) \;,
\end{align*}
each process repeatedly comparing its two operands and replacing the larger by their difference.
Here $\parallel$ denotes parallel composition under an \emph{interleaving} semantics:
a computation is any interleaving of the steps of the two processes, which interfere through the shared variable~$B$.
Whether this system is confluent is \emph{subtle}, and depends on the granularity of its steps.
If each comparison-and-update is one \emph{atomic} step --- for instance, guarded by a lock on the two operands --- 
then the system is confluent with the unique result $A=B=C=\gcd(A,B,C)$.
If instead the comparison and the assignment are \emph{two separate} steps, the system is \emph{not} confluent
as a consequence of lost updates through stale reads.
The outcome then depends on the schedule.
Confluence is thus a subtle property, and reasoning about it soundly calls for rigorous, formal methods.

Confluence is most elegantly studied at the abstract level of transition systems, 
also called \emph{abstract rewriting systems}~\cite{tere:2003,klop:1992}; 
criteria proved there apply uniformly to every concrete system that instantiates them.
Henrio et al.~\cite{henrio:etal:2026} have recently demonstrated that abstract confluence methods
--- in particular van~Oostrom's decreasing diagrams and de~Bruijn's weak diamond property~\cite{brui:1978} --- 
can be employed fruitfully to prove confluence of \emph{actual} programs:
they mechanize de~Bruijn's weak diamond method in the Rocq proof assistant and use it to establish confluence of programs in an actor language.
Decreasing diagrams are equally central to the automated confluence analysis of rewriting systems, 
where they underlie the rule-labeling criteria of provers such as \textsf{CSI}~\cite{zank:felg:midd:2011:tool};
confluence is, indeed, a highly active research area, 
with an annual International Workshop on Confluence and \emph{Confluence Competition}~\cite{midd:nage:shin:2019}.
These are precisely the methods studied in the present paper, and completeness of the method means that, in principle,
this single technique can certify the confluence of \emph{every confluent program}.

Because confluence is undecidable for many classes of systems, a rich arsenal of \emph{sufficient criteria} for confluence has been developed:
Newman's lemma~\cite{newm:1942},
the Hindley--Rosen lemma~\cite{hind:1964,rose:1973},
Huet's strong confluence~\cite{huet:1980},
and many others~\cite{tere:2003,klop:1992}.
The most versatile of these are de~Bruijn's weak diamond property~\cite{brui:1978}
and van~Oostrom's \emph{decreasing diagrams} technique~\cite{oost:1994b,oost:1994}.
Both infer global confluence from local peaks alone, by labeling the rewrite steps with elements of a well-founded order
and constraining, through these labels, how each peak may be joined.
The two criteria are in fact equal in strength~\cite{endr:klop:2013}, yet decreasing diagrams are often more convenient to apply:
they replace de~Bruijn's \emph{total} order and asymmetric peak condition with a \emph{partial} order and a single \emph{symmetric} diagram (Figure~\ref{fig:decreasing:diagram}).\footnote{%
  For instance, Hirokawa, Middeldorp, Sternagel and Winkler~\cite[Corollary~2.17]{hiro:midd:ster:wink:2019}
  prove the critical pair criterion for terminating rewrite systems based on prime critical pairs
  by labeling each step with its source term and ordering the labels by the rewrite relation~$\to^+$ itself,
  a well-founded but non-total order; de~Bruijn's criterion would first require linearising this order.
}
A system whose confluence can, in principle, be established in this way is called \emph{decreasing Church--Rosser}~\cite{oost:1994b,oost:1994}.
The power of the method is striking: essentially all of the classical criteria above arise as instances by a suitable choice of labels.
This makes the defining question of the method unavoidable: is it \emph{complete}?

\begin{quote}
  \emph{Does the Church--Rosser property imply the decreasing Church--Rosser property?}
\end{quote}

\noindent
This is Problem~56 of the RTA List of Open Problems, posed by van~Oostrom in 1993~\cite{rtaloop56,dershowitz:jouannaud:klop:1991}: 
can confluence always be \emph{localized} --- witnessed by a decreasing labeling of the one-step relation?
In other words: is every confluent system provably confluent by this single technique?

The problem was known to have a positive answer for all \emph{countable} systems~\cite{oost:1994b}; 
Endrullis, Klop, and Overbeek later sharpened the countable case, showing that just \emph{two} labels suffice there~\cite{endr:klop:over:2020}.
The countable barrier was first crossed by Ivanov~\cite{ivanov:2025}, 
who proved that every confluent system of cardinality at most $\aleph_1$ 
is decreasing Church--Rosser --- in fact, with three labels. 
He also showed that two labels do not suffice for all confluent systems of cardinality $\aleph_1$, 
so three labels are optimal at this cardinality.
Conditional completeness of decreasing diagrams under additional assumptions 
(e.g., weak normalization) was studied in~\cite{oost:1994,endr:klop:over:2018}.
The general question of unconditional completeness of decreasing diagrams remained open.

Uncountable systems are not a curiosity.
Although the programs of interest are typically countable, confluence questions arise just as naturally for genuinely uncountable systems
--- rewriting over the reals, hybrid and cyber-physical systems, and infinitary or probabilistic semantics.
These are not hypothetical: a hybrid system has an uncountable continuous state space and yields a confluent,
uncountable ARS of the kind studied in verification (Remark~\ref{rem:uncountable:ars}).
Their cardinalities readily exceed~$\aleph_1$: the reals have cardinality~$2^{\aleph_0}$, and collections of functions over them are larger still.
Here set theory sets the terms of the problem.
Ivanov's bound~\cite{ivanov:2025} reaches the continuum only under the Continuum Hypothesis~\cite{smul:1996}; without it, $2^{\aleph_0}$ may lie \emph{arbitrarily} far up the hierarchy of alephs (Easton's theorem~\cite{easton:1970}).
An unconditional treatment therefore requires completeness at \emph{every} cardinality.
Such systems arise readily: a sensor logger that records real-valued readings (see Example~\ref{ex:sensor:logger})
yields a confluent ARS of cardinality~$2^{\aleph_0}$ without a cofinal rewrite sequence,
for which completeness was previously known only under the Continuum Hypothesis;
if the readings are sets of reals, its cardinality $2^{2^{\aleph_0}}$ exceeds $\aleph_1$ outright.
\\[-1ex]

\noindent
\textbf{Contributions.}
We resolve Problem~56~\cite{rtaloop56,dershowitz:jouannaud:klop:1991} in full by proving that \emph{every} confluent abstract rewrite system is decreasing
Church--Rosser, with \emph{no cardinality restriction}.
Moreover, we show that \emph{three labels always suffice} (Theorem~\ref{thm:main}), 
and this is optimal, since two labels do not suffice~\cite{ivanov:2025}.
We thereby also establish the \emph{completeness of the weak diamond property} for confluence,
since de~Bruijn's weak diamond property and decreasing diagrams are equivalent.

To the best of our knowledge, these are the first unconditional completeness results for methods of proving confluence.
This outcome was not foregone: van~Oostrom conjectured that decreasing diagrams are \emph{incomplete}~\cite[Conjecture~2.3.31]{oost:1994},
and related questions about uncountable directed sets are independent of ZFC~\cite{todo:1985}.
Moreover, the proof reveals hidden structure in every connected confluent system:
a cofinal, locally almost deterministic subsystem, generalizing Klop's cofinal reductions~\cite{klop:1980} beyond the countable case.
\\[-1ex]

\noindent
\textbf{Outline.}
Section~\ref{sec:preliminaries} fixes notation and recalls the decreasing diagrams method.
Section~\ref{sec:tools} introduces the key tools employed throughout the paper.
Section~\ref{sec:regular} develops the \emph{fan} construction, 
which settles the question for every system whose cofinality cardinal $\cfc{\to}$ is regular (more generally, for all almost regular systems).
Section~\ref{sec:fans} develops the \emph{double fan} construction for the remaining case 
where $\cfc{\to}$ is singular.
Together with the two-label countable case (Lemma~\ref{lem:countable}), 
this assembles the full completeness result in Section~\ref{sec:singular:results}; Section~\ref{sec:conclusion} concludes with directions for future work.
\\[-1ex]

\noindent
\textbf{Formalization.}
To eliminate any doubt about a delicate transfinite argument, 
the entire development is independently machine-checked in \emph{two} proof assistants, Isabelle/HOL and Lean~\cite{moura:ullrich:2021};
both formalizations are archived on Zenodo~\cite{endr:ivan:raam:2026:zenodo}.
Both proofs use only the standard axioms of classical mathematics.\footnote{
  The proof in Lean uses propositional extensionality, the axiom of choice, and quotient soundness~\cite{moura:ullrich:2021}.
  The proof in Isabelle/HOL uses HOL axioms~\cite{nipk:wenz:paul:2002}.
}
The mathematical content --- the definitions, constructions, and proofs --- is the authors' own;
the formalizations were developed with the assistance of Anthropic's Claude, under the
authors' direction and review. The authors take full responsibility for the correctness and
presentation of this work.

By the proof-theoretic conservativity of Isabelle/HOL~\cite{kuncar2018safety}, 
the theorem \textit{dcr\_three\_\allowbreak of\_\allowbreak confluent} of \textit{Singular.thy}, with its definitions unfolded, 
is already provable in Pitts' standard HOL~\cite{pitts1993hol}, and hence our completeness result holds in ZFC.
A Lean proof alone would not give this, as Lean's type theory needs inaccessible cardinals to be interpreted in ZFC~\cite{carneiro:2024}.

\section{Preliminaries}\label{sec:preliminaries}
We recall some central definitions from abstract rewriting and the decreasing diagrams technique.
In the paper we work in ZFC~\cite{frae:bar:levy:1973}; in particular, we freely use the Axiom of Choice and standard facts about well-orderings, ordinals, and cardinals.

We use standard notation for cardinals.
For a cardinal $\kappa$, we write $\kappa^+$ for its \emph{successor cardinal}, the least cardinal greater than~$\kappa$.
For an infinite cardinal $\kappa$, its \emph{cofinality} $\cf{\kappa}$ is the least cardinal $\lambda$
such that $\kappa$ is the union of $\lambda$ sets, each of cardinality $< \kappa$.
This set-theoretic cofinality $\cf{\kappa}$ of a cardinal should not be confused with
the cofinality $\cfc{\to}$ of an ARS $\pair{A}{{\to}}$, introduced in Definition~\ref{def:cofinal:subset}.
An infinite cardinal $\kappa$ is \emph{regular} if $\cf{\kappa} = \kappa$, and \emph{singular} otherwise.
Every infinite successor cardinal is regular.
Hence every singular cardinal $\kappa$ is an uncountable limit cardinal, that is, $\beta < \kappa$ implies $\beta^+ < \kappa$.

Let $A$ be a set.
For a relation ${\to} \subseteq A\times A$ we write
\begin{itemize}
  \item $\to^*$ or $\steps$ for its reflexive transitive closure, and $\to^+$ for its transitive closure;
  \item $\lstep$ for the inverse of $\step$;
  \item $\cstep$ for the symmetric closure of $\step$;
  \item $\equiv$ to denote the identity relation, that is, ${\equiv} = \{(a,a) \mid a \in A\}$;
  \item ${\to^\equiv}$ for the reflexive closure ${\to \cup \equiv}$ of $\to$;
\end{itemize}
For a set $B \subseteq A$, we moreover define ${\step}(B) = \set{c \mid b \in B,\; b \step c }$, and we write $\size{B}$ for the cardinality of~$B$.
If $a \steps b$, then $b$ is a \emph{reduct} of $a$, and we say that $a$ \emph{reaches}~$b$.

\begin{definition}[Abstract Rewriting System]
  An \emph{abstract rewriting system (ARS)} $\pair{A}{{\to}}$ consists of a non-empty set~$A$
  and a binary relation ${\to} \subseteq A \times A$.
  An ARS $\pair{B}{{\substep}}$ is called a \emph{sub-ARS} (or \emph{subsystem}) of $\pair{A}{{\to}}$,
  denoted $\pair{B}{{\substep}} \subseteq \pair{A}{{\step}}$,
  if $B \subseteq A$ and ${\substep} \subseteq {\step}$.
\end{definition}

\noindent
Abstract rewriting systems are sometimes defined more generally, as a set $A$ of objects together
with a set of \emph{steps} and two functions assigning to each step a \emph{source} and a
\emph{target} in $A$.  In contrast to a relation, this permits several distinct \emph{parallel}
steps with the same source and target.  All results of this paper carry over to this more general
setting: it suffices to assign the same label to all parallel steps (steps sharing source and target),
which reduces the general case to the relational one.
In the remainder of the paper, we therefore work with ARSs whose steps form a relation ${\to} \subseteq A \times A$.

\begin{definition}[Rewrite sequence]
  A \emph{rewrite sequence} in $\pair{A}{{\to}}$ is a finite or infinite sequence
  $a_0 \to a_1 \to a_2 \to \ldots$ of elements of $A$ with $a_i \to a_{i+1}$ for every consecutive pair.
  It is called \emph{simple} if it visits no element twice ($i \ne j$ implies $a_i \ne a_j$).
\end{definition}

\noindent
An element $a \in A$ is a \emph{normal form} if there is no $b \in A$ with $a \to b$.

\begin{remark}\label{rem:uncountable:ars}
  Uncountable ARS arise naturally from hybrid (discrete--continuous) systems~\cite{goeb:2009}.
  With any such system $H$ we associate an ARS $\pair{S}{{\to}}$ whose states $S$ are the continuous states of $H$,
  and where $s_1 \to s_2$ holds whenever $H$ reaches $s_2$ from $s_1$ by continuous evolution within one discrete mode,
  or by a single (possibly jump-inducing) mode change.
  Since the continuous state space is real-valued, $S$ is typically uncountable.

  For instance, a bouncing ball~\cite[Figure 4]{ivanov:2023} has state $(y,v) \in \mathbb{R}^2$
  ($y$ the height above ground, $v$ the velocity) and two modes: ``in motion'', governed by
  $\frac{d}{dt}(y,v) = (v,-g)$ while $y > 0$ with gravity $g$, and ``stopped'' at $y = v = 0$;
  each floor impact ($y \to 0$) resets the velocity to $-cv$ for a fixed restitution coefficient~$c$.
  This standard model exhibits Zeno behavior~\cite{goeb:2009} --- infinitely many bounces in finite time ---
  and never reaches the ``stopped'' mode, so it fails to capture a ball that comes to rest.
  Following~\cite{lee:2016}, one repairs this with nondeterminism: the ``in motion'' ball may switch
  to ``stopped'' once $|y|$ and $|v|$ are small enough, making that mode reachable.
  The resulting ARS $\pair{S}{{\to}}$ is confluent (the ``stopped'' mode is always reachable),
  nondeterministic, non-terminating, and uncountable.
  Here confluence means that the outcome does not depend on the nondeterministic timing of the mode switch:
  whichever switch is taken, the ball comes to rest in the same state.
  As this resting state is reachable from every state, the system is a trivial instance of our results (see Convention~\ref{conv:components});
  the non-trivial cases are illustrated by the running Example~\ref{ex:successor} and by Examples~\ref{ex:finite:reals} and~\ref{ex:sensor:logger}.
  Further examples, combining the nondeterminism of physical models with that of control software in
  cyber-physical systems, follow from~\cite{lee:2016,goeb:2009,plat:2010}.
\end{remark}

\begin{definition}[Confluence]
  An ARS $(A,\to)$ is \emph{confluent (\CR)} if ${\lsteps \cdot \steps} \subseteq {\steps \cdot \lsteps}$,
  that is, every pair of finite, coinitial rewrite sequences can be joined to a common reduct. 
\end{definition}

\begin{definition}[Deterministic]\label{def:deterministic}
  Let $\pair{A}{{\step}}$ be an ARS.
  An element $a \in A$ is called \emph{deterministic} if, for all $b,c$ with $a \step b$ and $a \step c$, we have $b = c$.
  We say that a step $a \to b$ is \emph{deterministic} if its source $a$ is deterministic,
  and we call a rewrite sequence \emph{deterministic} if all steps in the sequence are deterministic.
  The ARS $\pair{A}{{\step}}$ is called \emph{deterministic} if all its elements are deterministic.
\end{definition}

\begin{definition}[Almost deterministic]\label{def:almost:deterministic}
  Let $\pair{A}{{\step}}$ be an ARS.
  A rewrite sequence is \emph{almost deterministic} if it contains at most one non-deterministic step.
  The ARS $\pair{A}{{\step}}$ is
  \begin{enumerate}
    \item 
      \emph{locally almost deterministic} if whenever $a \lstep \cdot \step b$
      there exist almost deterministic reductions $a \steps c$ and $b \steps c$ to some $c \in A$, and
    \item 
      \emph{globally almost deterministic} if whenever $a \cstep^* b$
      there exist almost deterministic reductions $a \steps c$ and $b \steps c$ to some $c \in A$.
  \end{enumerate}
\end{definition}

\noindent
Intuitively, an almost deterministic reduction makes at most one genuine choice; 
this is exactly what a two-label decreasing diagram can absorb (see Lemma~\ref{lem:almost:deterministic}).

\begin{remark} 
  The concepts of globally/locally almost deterministic ARS are analogous to, 
  but somewhat less restrictive than the concept of a 1-branching one-step nondeterministic 
  (or, alternatively, almost deterministic) Church-Rosser strategy considered in~\cite{ivanov:2026a}, 
  where it was shown that every confluent ARS with reduction relation of the cardinality $\aleph_1$ has a 1-branching 1-nCR strategy. 
\end{remark}

\begin{definition}[Connected component]\label{def:connected:component}
  Two elements $a,b \in A$ are \emph{connected} if $a \cstep^* b$.
  A set $C \subseteq A$ is \emph{connected} if, for every $a,b \in C$, we have $a \cstep^* b$.
  Moreover, $C$ is called a \emph{connected component} of $\pair{A}{{\step}}$ if it is a maximal connected subset of $A$.
\end{definition}

\begin{definition}[Cofinal subset]\label{def:cofinal:subset}
  A set $C \subseteq A$ is a \emph{cofinal subset} of $\pair{A}{{\to}}$
  if, for every $a \in A$, there exists $c \in C$ such that $a \to^* c$.
  We define
  \begin{align}
    \cfc{\to} = \min \set{ \size{C} \mid \text{$C$ is a cofinal subset of $\pair{A}{{\to}}$}}\;,
  \end{align}
  that is, the minimal cardinality of any cofinal subset of $\pair{A}{{\to}}$.
  A subsystem $\pair{B}{{\substep}}$ of $\pair{A}{{\to}}$ is \emph{cofinal} if its carrier $B$ is a cofinal subset of $\pair{A}{{\to}}$.
\end{definition}

\begin{example}[Successor relation]\label{ex:successor}
  Consider $\pair{\nat}{{\to}}$ with $n \to n+1$ for all $n \in \nat$;
  this ARS is connected and confluent (even deterministic).
  A set $C \subseteq \nat$ is cofinal if and only if it is infinite, so $\cfc{\to} = \aleph_0$.
\end{example}

\begin{remark}\label{rem:finite:cfc}
  If $\pair{A}{{\step}}$ is confluent and connected, and $\cfc{\step}$ is finite, then $\cfc{\step} = 1$.
  Indeed, let $C = \set{c_1,\ldots,c_n}$ be a finite cofinal subset.
  Any two elements of $C$ are connected and hence, by confluence, have a common reduct;
  by induction on $n$, all elements of $C$ have a common reduct~$d$.
  Then $\set{d}$ is cofinal, and hence $\cfc{\step} = 1$.
  Without confluence this fails:
  the ARS with steps $a \to b$, $a \to c$, $c \to d$, $c \to e$ is connected,
  but its three normal forms $b, d, e$ must belong to every cofinal subset, so $\cfc{\step} = 3$.
\end{remark}

\begin{notation}\label{not:family}
  For a family $\family{\to_\alpha}_{\alpha\in I}$ of relations $\to_\alpha \subseteq A \times A$ 
  indexed by a well-founded partial order $\pair{I}{{<}}$, we define
  \begin{align*}
    {\to} \;&=\; \textstyle{\bigcup_{\alpha\in I} \to_\alpha} &
    {\to_{< \beta}} \;&=\; \textstyle{\bigcup_{\alpha < \beta} \to_\alpha} &
    {\to_{\le \beta}} \;&=\; \textstyle{\bigcup_{\alpha \le \beta} \to_\alpha}  
  \end{align*}
  Moreover, we use $\to_{< \alpha \cup < \beta}$ as shorthand for $(\to_{< \alpha} \cup \to_{< \beta})$.
\end{notation}

The \emph{decreasing diagrams} method of van~Oostrom~\cite{oost:1994b}
derives confluence from a well-founded labeling of the rewrite steps,
by requiring that every \emph{peak} can be completed into a \emph{decreasing} elementary diagram
as shown in Figure~\ref{fig:decreasing:diagram}:
the peak $c \lstep_\beta a \to_\alpha b$ is joined at a common reduct~$d$ by
$b \to_{<\alpha}^{*} \cdot \to_{\beta}^{\equiv} \cdot \to_{<\alpha\,\cup\,<\beta}^{*} d$ and
$c \to_{<\beta}^{*} \cdot \to_{\alpha}^{\equiv} \cdot \to_{<\alpha\,\cup\,<\beta}^{*} d$,
where $\to_{<\alpha}$ collects the steps whose label is strictly below~$\alpha$ (in the partial order),
and $\to_{\gamma}^{\equiv}$ denotes at most one step of label~$\gamma$.

\begin{figure}[h!]
  \centering
  \includegraphics{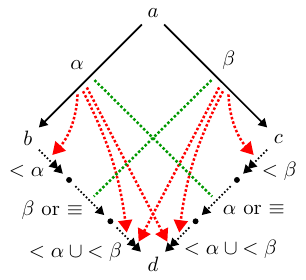}
  \Description{Decreasing elementary diagram.}
  \vspace{-1.5ex}
  \caption{Decreasing elementary diagram.
    The inner dotted arrows are just a visual aid:
    the dotted green \emph{lines} stand for weak decrease ($\ge$),
    the dotted red \emph{arrows} for strict decrease ($>$).
    Furthermore, multiple incoming arrows represent choice, thus weakening the requirements.
  }
  \label{fig:decreasing:diagram}
\end{figure}

\begin{definition}[Decreasing Church--Rosser]\label{def:decreasing:diagrams}
  An ARS $\pair{A}{{\to}}$ is \emph{decreasing Church--Rosser (\DCR) over a well-founded partial order $\pair{I}{{<}}$}
  if ${\to} = \bigcup_{\alpha\in I} {\to_\alpha}$ for some family $\family{\to_\alpha}_{\alpha\in I}$ of relations $\to_\alpha \subseteq A \times A$
  (see Notation~\ref{not:family}) such that
  every (local) peak $c \lstep_\beta a \to_\alpha b$ can be joined
  by rewrite sequences of the form shown in Figure~\ref{fig:decreasing:diagram}.\footnote{%
    Van Oostrom~\cite{oost:2008} generalizes the shape of the decreasing elementary diagrams by allowing the joining reductions to be conversions.
    This can be helpful to find suitable elementary diagrams.
    Since reductions are special conversions, our completeness result for the version with reductions also yields completeness for the version with conversions.
  }
  Such a family $\family{\to_\alpha}_{\alpha\in I}$ is called a \emph{locally decreasing labeling} of $\to$.
  For an ordinal $\gamma$, we write $\DCR_\gamma$ for decreasing Church--Rosser over $\pair{\set{\beta \mid \beta < \gamma}}{<}$;
  in particular, $\DCR_n$ for $n \in \nat$ means that the $n$ labels $0, \ldots, n-1$ suffice.
\end{definition}

The following is the main theorem of decreasing diagrams~\cite{oost:1994b}.
\begin{theorem}[Soundness of Decreasing Diagrams]\label{thm:decreasing:diagrams}
  If an ARS is decreasing Church--Rosser,
  then it is confluent. \qed
\end{theorem}
\noindent
In other words $\DCR \implies \CR$.

\section{Key concepts and tools}\label{sec:tools}

We begin by sketching the overall strategy.
Since a locally decreasing labeling can be assembled component by component, we may restrict attention to a single connected, confluent ARS.
Within a component, our goal is to carve out a \emph{cofinal} subsystem $\pair{S}{{\substep}}$ that is \emph{locally almost deterministic} (Definition~\ref{def:almost:deterministic}):
every local peak is joined by two reductions, each making at most one nondeterministic choice.
Such a subsystem is decreasing with only \emph{two} labels (Lemma~\ref{lem:almost:deterministic}),
and a decreasing labeling of a cofinal subsystem extends to the whole system at the cost of a single extra label (Lemma~\ref{lem:extend}) ---
yielding the \emph{three} labels of our main theorem, which is optimal.
The entire difficulty thus lies in constructing such a cofinal, locally almost deterministic subsystem, and this is what the tools of this section serve.
The notions of \emph{trap}, \emph{escape}, and \emph{trap degree} measure how large a subsystem a point can always escape.
In Section~\ref{sec:regular}, we employ \emph{fans} that are distinguished points of maximal escape ability;
they act as waypoints for joining the subsystem together as it grows.

We now introduce these tools in detail.

\begin{convention}[Connected components]\label{conv:components}
  For finding a locally decreasing labeling, it suffices to consider each connected component separately:
  every local peak lies within a single component,
  so locally decreasing labelings of the components together form a locally decreasing labeling of the whole system.
  Moreover, an ARS is confluent if and only if all its connected components are.
  Therefore, in Sections~\ref{sec:tools}--\ref{sec:fans}, we mainly consider ARSs $\pair{A}{{\to}}$ that are
  \emph{connected} and \emph{confluent}, and satisfy $\cfc{\to} > 1$ (see Definition~\ref{def:cofinal:subset}).
  The excluded case $\cfc{\to} = 1$ is trivial: then some element $m$ is reachable from every element,
  and two labels suffice (Lemma~\ref{lem:countable}).
  Under these assumptions, there are no normal forms:
  a normal form would be reachable from every element, and thus form a cofinal singleton.
\end{convention}

\subsection{Cofinal subsystems}

The concept of cofinal subsets plays a prominent role in our constructions.
In Sections~\ref{sec:regular} and~\ref{sec:fans}, 
we will construct a connected, confluent subsystem $\pair{S}{{\substep}}$
whose carrier $S$ is a cofinal subset of~$\pair{A}{{\step}}$.
A locally decreasing labeling of such a subsystem extends to the whole system at the cost of a
single additional label (Lemma~\ref{lem:extend}).
For two labels, this is the construction underlying the countable-completeness result of Endrullis,
Klop and Overbeek~\cite{endr:klop:over:2018,endr:klop:over:2020}; the general formulation, as a step
from $\DCR_n$ to $\DCR_{n+1}$ is due to Ivanov~\cite{ivanov:2025}.

\begin{lemma}[Labeling extension]\label{lem:extend}
  Let $\pair{B}{{\substep}}$ be a connected, confluent sub-ARS of $\pair{A}{{\step}}$
  such that $B$ is a cofinal subset of $\pair{A}{{\step}}$.
  If $\pair{B}{{\substep}}$ is decreasing Church--Rosser over 
  $\pair{\set{\beta \mid \beta < \alpha}}{<}$ for some ordinal $\alpha \ge 1$, 
  then $\pair{A}{{\step}}$ is decreasing Church--Rosser over $\pair{\set{\beta \mid \beta \le \alpha}}{<}$.
\end{lemma}

\begin{proof}
  We extend the labeling from $\pair{B}{{\substep}}$ to $\pair{A}{{\step}}$ as follows.
  For each node $a \in A \setminus B$, 
  we choose one step $a \to a'$ that starts a shortest path from $a$ to $B$ (so $a'$ is one step closer to $B$), and label this step with label $0$.
  For all steps of $\pair{A}{{\step}}$ that remain unlabeled, we assign label $\alpha$.
  Now consider a peak $c \lstep_\gamma x \to_\delta b$ with $x \in A$.
  If this peak is inside the subsystem $\pair{B}{{\substep}}$, then there is a decreasing elementary diagram inside the subsystem.
  Otherwise, the peak is not (fully) inside the subsystem.
  If $b = c$, we can join with empty steps.
  If $b \neq c$, then one of the steps is labeled by $\alpha$, that is, $\gamma = \alpha$ or $\delta = \alpha$.
  Consequently the joining reductions can employ arbitrary steps with labels $< \alpha$:
  we can first use steps with label $0$ to reach the subsystem $\pair{B}{{\substep}}$,
  and then use steps with labels $< \alpha$ to join inside the confluent subsystem.
\end{proof}

In~\cite{endr:klop:over:2018} it has been established that 
countable systems can be labeled with $2$ labels.
A related lemma on the existence of a spanning forest in a countable confluent ARS was  also given in~\cite[Lemma 1]{hiro:nage:oost:oyam:2019}.
In~\cite{ivanov:2025}, Ivanov has extended this result to all systems with a countable cofinality $\cfc{\step}$.
We call a rewrite sequence \emph{cofinal} in $\pair{A}{{\to}}$ if the set of its elements is a cofinal subset of $\pair{A}{{\to}}$.

\begin{lemma}[Countable cofinality]\label{lem:countable}
   Every connected confluent ARS $\pair{A}{{\step}}$ with countable cofinality $\cfc{\step}$
   has a cofinal rewrite sequence
   and is decreasing Church--Rosser over $\pair{\set{0,1}}{<}$. 
\end{lemma}

\begin{proof}
  Assume that $\pair{A}{{\to}}$ is connected, confluent, and has countable cofinality. 
  Let us show the existence of a cofinal rewrite sequence. Consider two cases.
  
  \textit{Case 1}: $\cfc{\to}$ is finite. 
  Then $\cfc{\to} = 1$ by Remark~\ref{rem:finite:cfc}, so $\pair{A}{{\to}}$ has some element $m \in A$ that is reachable from every other element of $A$. 
  Then the one-element sequence $(m)$ is a (finite) cofinal rewrite sequence.
  
  \textit{Case 2}: $\cfc{\to}$ is infinite. 
  Let $B$ be an infinite countable cofinal subset of $\pair{A}{{\to}}$ and $f$ be a choice function such that for each $a \in B$, $f(a) \in \{b~|~a\,\leftrightarrow\,b\}$. Let $C = B \cup f(B)$, where $f(B)$ is the image of a set under $f$. Then $C$ is a countable \textit{subfield}~\cite{ivanov:2025} of $\to$, i.e., the field of some countable subrelation of $\to$. Then by~\cite[Lemma 24]{ivanov:2025} (a lemma about connected Church--Rosser subfield extension) applied to the relation $\to$ and the set $C$ (and arbitrary sets $B,B',Ps$ and element $x$ in the notation of~\cite{ivanov:2025}), $C$ has an extension to a countable subset $C'$ with $C \subseteq C' \subseteq A$ such that $\pair{C'}{{\to|_{C'}}}$ is connected and confluent. Note that any countable confluent ARS has Klop's cofinality property~\cite[Definition 7 and Theorem 9]{endr:klop:over:2018}. Then since $\pair{C'}{{\to|_{C'}}}$ is connected, from~\cite[Lemma 8]{endr:klop:over:2018} it follows that $\pair{C'}{{\to|_{C'}}}$ has a cofinal rewrite sequence that is also cofinal in $\pair{A}{{\to}}$.
  
  Thus, in any case, $\pair{A}{{\to}}$ has some nonempty finite or infinite cofinal rewrite sequence $a_1 \to a_2 \to a_3 \to \ldots$. We may assume that this sequence is simple: if some element occurs infinitely often, then it is reachable from every element, so $\cfc{\to} = 1$ and Case~1 applies; otherwise, erasing loops (repeatedly continuing from the last occurrence of the current element) yields a simple cofinal rewrite sequence. If this sequence is considered as a sub-ARS of $\pair{A}{{\to}}$, it is connected, confluent, and, moreover, decreasing Church--Rosser over $\pair{\set{0}}{<}$ (all rewrite steps in it can be labeled with 0 for the purpose of decreasing diagrams). Thus $\pair{A}{{\to}}$ is decreasing Church--Rosser over $\pair{\set{0,1}}{<}$ by Lemma~\ref{lem:extend}.
\end{proof}

\begin{example}[Example~\ref{ex:successor} continued]
  The system $\pair{\nat}{{\to}}$ is itself the cofinal rewrite sequence $0 \to 1 \to 2 \to \cdots$ of Lemma~\ref{lem:countable}.
\end{example}

\subsection{Traps and escapes}

The following concepts of \emph{escape}, \emph{trap}, and \emph{trap degree} have been pioneered in~\cite{ivanov:2025}.
They are also key tools for our constructions.
Informally, a high trap degree means ``room to maneuver'': the greater a point's trap degree, the larger the subsystems it can always escape.

\begin{definition}[Escape]
  An \emph{escape} of an element $a \in A$ from $T \subseteq A$ is 
  a finite rewrite sequence $a \to a_1 \to \ldots \to a_n$ ($n \ge 1$) 
  such that $a_1, \ldots, a_n \not\in T$ and, for every $t \in T$, it holds that $a_n \not\to^* t$.
  This concept is illustrated in Figure~\ref{fig:escape}.
  Without loss of generality, we will tacitly assume that escape paths are simple
  (this can always be achieved by cutting out cycles).
\end{definition}
\noindent
Observe that the first element $a$ of the escape is allowed to be in $T$. 
In particular, $a$ can escape $\set{a}$ iff $a \to^+ b$ for some $b$ with $b \not\to^* a$.

\begin{example}[Example~\ref{ex:successor} continued]
  An element $n \in \nat$ has an escape from $T \subseteq \nat$ if and only if $T$ contains no element $> n$:
  an escape $n \to n+1 \to \cdots \to m$ must avoid $T$ and end in some $m$ with $m > t$ for all $t \in T$.
  For instance, the single step $n \to n+1$ is an escape of $n$ from $\set{0,\ldots,n}$.
\end{example}

\begin{lemma}[Escape or reach]\label{lem:escape:or:reach}
  Let $\pair{A}{{\to}}$ be connected and confluent such that $\cfc{\to} > 1$,
  and let $S \subseteq A$ and $a \in A$.
  If $a$ does not reach $S$ (that is, $a \not\to^* s$ for all $s \in S$),
  then $a$ has an escape from $S \cup \set{a}$, and hence also from $S$.
\end{lemma}

\begin{proof}
  Since $\cfc{\to} > 1$, the set $\set{a}$ is not cofinal, so there is some $c \in A$ with $c \not\to^* a$.
  By connectedness and confluence, $a$ and $c$ have a common reduct $d$.
  Then $d \not\to^* a$, since otherwise $c \steps d \steps a$; in particular $d \ne a$.
  Let $a = a_0 \to a_1 \to \cdots \to a_n = d$ be a simple rewrite sequence; then $n \ge 1$.
  For $1 \le i \le n$ we have $a_i \ne a$ (by simplicity) and $a_i \notin S$ (since $a$ does not reach $S$).
  Moreover, $d$ reaches neither $a$ nor any element of $S$.
  Thus the sequence is an escape of $a$ from $S \cup \set{a}$,
  and every escape from a set is also an escape from each of its subsets.
\end{proof}

\begin{figure}[ht!]
  \includegraphics{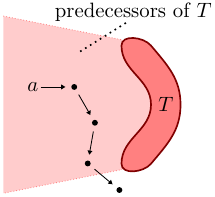}
  \caption{An escape of $a$ from $T$: 
    the rewrite sequence avoids $T$ and ends outside of the predecessors of $T$.
    (The element $a$ itself is allowed to be part of $T$).
    If $a$ cannot escape $T$, then $T$ is called a trap for $a$.
    }
  \label{fig:escape}
\end{figure}

\begin{definition}[Trap]
  A set $T \subseteq A$ \emph{traps} $A' \subseteq A$ if no $a \in A'$ has an escape from $T$.
  Then we also say that $T$ is a \emph{trap for $A'$}.
\end{definition}

\begin{example}[Example~\ref{ex:successor} continued]
  By the characterization of escapes above, a set $T \subseteq \nat$ is a trap for $n$ if and only if $T$ contains an element $> n$.
  For instance, $\set{n+1}$ is a trap for $n$, while $\emptyset$ and $\set{0,\ldots,n}$ are not.
  Consequently, $T$ traps all of $\nat$ if and only if it contains elements above every $n$, that is, if and only if $T$ is infinite.
\end{example}
  
\begin{definition}[Trap degree]\label{def:trap:degree}
  The \emph{trap degree} of a set of elements $A' \subseteq A$ is 
  \begin{align}
    \trap{A'} = \min \set{ \size{B} \mid \text{$B$ is a trap for $A'$}},
  \end{align}
  that is, the minimal cardinality of any trap $B$ for $A'$.
  For $a \in A$, we also write $\trap{a}$ as shorthand for $\trap{\set{a}}$.
\end{definition}

\begin{example}[Example~\ref{ex:successor} continued]
  We have $\trap{n} = 1$ for every $n \in \nat$, whereas $\trap{\nat} = \aleph_0 = \cfc{\to}$.
\end{example}

\noindent
The crucial property employed throughout our constructions is that an element with trap degree $\alpha$
can escape any subsystem $\pair{S}{{\substep}}$ for which $\size{S} < \alpha$.

The trap degree $\trap{A}$ of the entire system is closely related to the minimal cardinality of a cofinal subset.
The difference between a trap for $A$ and a cofinal subset of $\pair{A}{{\to}}$ only concerns the normal forms.
The normal forms must be included in any cofinal subset, but they don't need to be part of a trap (they are already trapped by themselves).
In the setting of our construction (considering a single connected and confluent component) the two coincide: 

\begin{lemma}\label{lem:trap:cfc}
  If $\pair{A}{{\to}}$ is connected and confluent, and $\cfc{\to}$ is infinite,
  then $\trap{A} = \cfc{\to}$. 
\end{lemma}

\begin{proof}
  There are no normal forms at all (a normal form $n$ would be reachable from every element by confluence and connectedness, giving $\cfc{\to} = 1$), 
  and without normal forms every cofinal subset is a trap and vice versa, so $\trap{A} = \cfc{\to}$. 
\end{proof}

\noindent
Every trap for $A$ is also a trap for each single element $a \in A$.
Hence $\trap{a} \le \trap{A}$ for every $a \in A$, and in the setting of Lemma~\ref{lem:trap:cfc}
the cardinal $\cfc{\to}$ is an upper bound for the trap degree of all elements.
This justifies the following terminology.

\begin{definition}[Maximal trap degree]\label{def:maximal:trap:degree}
  Let $\pair{A}{{\to}}$ be connected and confluent such that $\cfc{\to}$ is infinite.
  An element $a \in A$ has \emph{maximal trap degree} if $\trap{a} = \cfc{\to}$ $(= \trap{A})$.
\end{definition}

\begin{example}[Example~\ref{ex:successor} continued]
  No element of $\nat$ has maximal trap degree, since $\trap{n} = 1 < \aleph_0 = \cfc{\to}$.
  The next example shows the other extreme.
\end{example}

\begin{example}[Finite sets of reals]\label{ex:finite:reals}
  Let $A$ be the set of finite subsets of $\mathbb{R}$, with steps $X \to X \cup \set{r}$ for every $r \in \mathbb{R} \setminus X$.
  This ARS is connected and confluent: any $X, Y \in A$ have the common reduct $X \cup Y$.
  Let $C \subseteq A$ with $\size{C} < 2^{\aleph_0}$. Then $\size{\bigcup C} < 2^{\aleph_0}$, so there is some $r \in \mathbb{R} \setminus \bigcup C$;
  the element $\set{r}$ only reaches sets containing $r$, none of which is in $C$.
  Hence no such $C$ is cofinal, and $\cfc{\to} = \size{A} = 2^{\aleph_0}$.
  The same argument shows that every $X \in A$ has maximal trap degree:
  for every $T \subseteq A$ with $\size{T} < 2^{\aleph_0}$, choose $r \notin X \cup \bigcup T$;
  then the single step $X \to X \cup \set{r}$ is an escape of $X$ from $T$,
  since $X \cup \set{r} \notin T$ and all reducts of $X \cup \set{r}$ contain $r$.
  So $\trap{X} = 2^{\aleph_0}$ for every $X \in A$.
  In particular, this system is almost regular (Definition~\ref{def:almost:regular}),
  regardless of whether $2^{\aleph_0}$ is a regular cardinal, and it is covered by the construction of Section~\ref{sec:regular}.
  The system has no cofinal rewrite sequence (the elements of a rewrite sequence have a countable union),
  and since $\size{A} = 2^{\aleph_0}$, the earlier completeness result for cardinality $\aleph_1$~\cite{ivanov:2025} covers it only under the Continuum Hypothesis.
\end{example}

\begin{figure}
  \centering
  \includegraphics{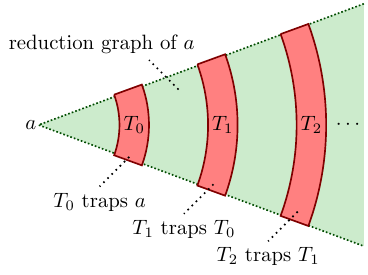}
  \caption{The reduction graph of $a$ with iterated traps along the way.
    $T_0$ is a minimal trap for $a$ and, for $n \ge 0$, $T_{n+1}$ is a minimal trap for $T_n$.}
  \Description{Reduction graph of a with iterated traps.}
  \label{fig:river}
\end{figure}

\noindent
For $a \in A$, the \emph{reduction graph} of $a$ is the sub-ARS of $\pair{A}{{\to}}$ 
consisting of the reducts $\set{b \mid a \steps b}$ of $a$ together with the restriction of $\to$ to this set;
the reduction graph of a set $C \subseteq A$ is the union of the reduction graphs of its elements.

\begin{lemma}[Iterated traps]\label{lem:iterated:traps}
  Let $\pair{A}{{\to}}$ be connected and confluent such that $\cfc{\to}$ is infinite.
  Let $a \in A$.
  Let $T_0$ be a trap of minimal cardinality for $a$,
  and, for every $n \in \nat$, let $T_{n+1}$ be a trap of minimal cardinality for $T_n$.
  This construction is illustrated in Figure~\ref{fig:river}.
  Then $T = \bigcup_{n \in \nat} T_n$ is a cofinal subset of $\pair{A}{{\to}}$.
\end{lemma}

\begin{proof}
  No reduct $b$ of $a$ can escape $T$: otherwise, consider a reduction from $a$ via $b$ along the escape.
  If it contains no point of $T$, it is an escape of $a$ from $T_0$;
  otherwise its part after the last point $p \in T_n$ of $T$ is an escape of $p$ from $T_{n+1}$.
  Both contradict the choice of the traps.
  As there are no normal forms (by the proof of Lemma~\ref{lem:trap:cfc}), a point that cannot escape $T$ reaches $T$ (its first step would otherwise be an escape).
  Hence every reduct of $a$ reaches $T$, and since every element has a common reduct with $a$ by connectedness and confluence, $T$ is cofinal.
\end{proof}

\begin{example}[Example~\ref{ex:successor} continued]
  For $a = n$, the iterated traps can be chosen as $T_k = \set{n+k+1}$; their union $\set{m \mid m > n}$ is indeed cofinal.
\end{example}

\section{Completeness for regular cardinals}\label{sec:regular}

We prove that every connected and confluent ARS $\pair{A}{{\to}}$ with regular cofinality cardinal $\cfc{\step}$
is decreasing Church--Rosser over $\pair{\set{0,1,2}}{<}$;
the singular case is treated in Section~\ref{sec:fans}.
The regularity of $\cfc{\step}$ ensures, for uncountable $\cfc{\step}$, a cofinal supply of points of maximal trap degree (Lemma~\ref{lem:trap:degree}).

\begin{lemma}[Trap degree]\label{lem:trap:degree}
  Let $\pair{A}{{\to}}$ be confluent and connected such that $\cfc{\step}$ is uncountable.
  For every element $a \in A$ and cardinal $\alpha < \cfc{\step}$, there exists $b \in A$ such that $a \steps b$ and $\trap{b} \ge \alpha$.
  If $\cfc{\step}$ is regular, the same holds for $\alpha = \cfc{\step}$;
  then $\trap{b} = \cfc{\step}$, that is, $b$ has maximal trap degree.
\end{lemma}

\begin{proof}
  Let $a \in A$ and $\alpha < \cfc{\step}$ be arbitrary.
  For $n \in \nat$, let $T_n$ be defined as in Lemma~\ref{lem:iterated:traps}; see also Figure~\ref{fig:river}.
  We may choose all $T_n$ inside the reduction graph of $a$: 
  if $T$ traps $X$, then so does the subset of $T$ consisting of the reducts of elements of $X$, 
  since an escape from this subset only visits and reaches such reducts, and is therefore also an escape from $T$.
  For a contradiction, assume that $\trap{b} < \alpha$ for all $b \in A$ with $a \steps b$. 
  Without loss of generality, we may assume that $\alpha \ge \omega$.
  Then it follows that $\size{T_n} \le \alpha$ for every $n \in \nat$; we prove this by induction on $n$.
  By assumption $\size{T_0} = \trap{a} \le \alpha$, and for $n \ge 0$ we have:
  \begin{align}
    \size{T_{n+1}} = \trap{T_{n}} \le \sum_{c \in T_n} \underbrace{\trap{c}}_{\le \alpha} \overset{\text{IH}}{\le} \alpha \cdot \alpha = \alpha
  \end{align}
  since $\size{T_n} \le \alpha$ by the induction hypothesis (IH).
  Then also $\size{T} = \size{\bigcup_{n \in \nat} T_n} \le \omega \cdot \alpha \le \alpha$ for infinite $\alpha$.
  But $T$ is itself a cofinal subset of $\pair{A}{{\to}}$ by Lemma~\ref{lem:iterated:traps}.
  Hence $\cfc{\step} \le \size{T} \le \alpha$ which contradicts $\alpha < \cfc{\step}$.

  Now consider the case that $\cfc{\to}$ is regular and $\alpha = \cfc{\to}$.
  For a contradiction, assume that $\trap{b} < \alpha$ for all $b \in A$ with $a \steps b$. 
  For $n \in \nat$, we construct $T_n$ as above.
  By induction on $n \in \nat$ we prove that $\size{T_n} < \cfc{\step}$.
  We have $\size{T_0} = \trap{a} < \alpha$ by assumption, and for every $n \in \nat$:
  \begin{align}
    \size{T_{n+1}} = \trap{T_{n}} \le \sum_{c \in T_n} \underbrace{\trap{c}}_{< \alpha} \overset{\text{IH}}{<} \alpha
  \end{align}
  since $\size{T_n} < \alpha$ by the induction hypothesis and $\alpha$ is regular by assumption.
  Here we use that the sum of fewer than $\alpha$ many cardinals, each below~$\alpha$, is again below~$\alpha$ (since $\alpha$ is regular and uncountable).
  Then also $\size{T} = \size{\bigcup_{n \in \nat} T_n} < \alpha$, leading to a contradiction as above.
  Hence there is a reduct $b$ of $a$ with $\trap{b} \ge \cfc{\step}$,
  and $\trap{b} \le \trap{A} = \cfc{\step}$ holds for every $b \in A$ (see the paragraph before Definition~\ref{def:maximal:trap:degree}),
  so $\trap{b} = \cfc{\step}$.
\end{proof}

In fact, the construction presented in this section applies more broadly than systems with regular $\cfc{\step}$;
it applies to all \emph{almost regular} systems:

\begin{definition}[Almost regular]\label{def:almost:regular}
  A connected and confluent ARS $\pair{A}{{\to}}$ is called \emph{almost regular}
  if $\cfc{\step}$ is regular, or
  every element $a \in A$ reduces to some $b$ with $\trap{b} = \cfc{\to}$.
\end{definition}

\noindent
Note that the latter property in the definition is fulfilled by 
all systems with uncountable regular cofinality cardinal $\cfc{\step}$ by Lemma~\ref{lem:trap:degree}.

\begin{example}[Sensor logger]\label{ex:sensor:logger}
  The following system is beyond previous completeness results.
  Consider a sensor logger with states $(i,X)$, where $i \in \set{0,1,2}$ is the mode of operation
  and $X$ is the finite set of real numbers logged so far.
  After an initial two-point calibration (mode~$0$), the device records one reading at a time (mode~$2$),
  each followed by a deterministic processing step (mode~$1$).
  Its steps are given by the following rules, with $y, y_1, y_2$ ranging over $\mathbb{R}$:
  \begin{align*}
    (0,X) &\to (1, X \cup \set{y_1,y_2}) \quad \text{if } y_1 \ne y_2 && \text{(calibration)}\\
    (1,X) &\to (2, X) && \text{(processing)}\\
    (2,X) &\to (1, X \cup \set{y}) && \text{(reading)}
  \end{align*}
  Every state $(i,X)$ reduces to $(1,Z)$ for every finite $Z \supseteq X$ with $\size{Z} \ge 2$.
  Hence any two states have a common reduct, and the system is connected and confluent.
  As in Example~\ref{ex:finite:reals}, a fresh real~$r$ shows that no set of fewer than $2^{\aleph_0}$ states is cofinal, so $\cfc{\to} = 2^{\aleph_0}$.
  Every state $(2,X)$ has maximal trap degree $2^{\aleph_0}$:
  for $\size{T} < 2^{\aleph_0}$, a reading of a fresh real $y \notin X \cup \bigcup\set{Y \mid (j,Y) \in T}$ escapes $T$.
  In contrast, $\trap{(1,X)} = 1$, since the only successor $(2,X)$ forms a trap.
  As every state reaches mode~$2$, the system is almost regular,
  and Theorem~\ref{thm:regular} shows that it is decreasing Church--Rosser with three labels, without assuming the Continuum Hypothesis.
  The system has no cofinal rewrite sequence, and the earlier result for cardinality $\aleph_1$~\cite{ivanov:2025} covers it only under the Continuum Hypothesis.
  If the readings are sets of reals instead (for instance, regions computed by a reachability analysis),
  the same analysis applies with $2^{2^{\aleph_0}}$ in place of $2^{\aleph_0}$;
  then $\size{A} = 2^{2^{\aleph_0}} > \aleph_1$, and~\cite{ivanov:2025} does not apply at all.
\end{example}

For the remainder of this section, we assume that \emph{$\pair{A}{{\to}}$ is almost regular}, thus also connected and confluent.
The case of countable $\cfc{\step}$ has already been settled by Ivanov in~\cite{ivanov:2025} 
(see Lemma~\ref{lem:countable}, which uses only two labels), so we may assume that $\cfc{\step}$ is an \emph{uncountable} cardinal.

\paragraph{Outline of the construction:}
We enumerate a cofinal set $C$ of points of maximal trap degree $\cfc{\to}$ (Definition~\ref{def:maximal:trap:degree}) and process it by transfinite induction.
At each stage we add the next point as a \emph{fan} (a point of maximal trap degree at which the subsystem may branch, see Section~\ref{sec:regular:subsystem}) and join it pairwise to the fans already present.
Because a fan has trap degree $\cfc{\to}$ while the subsystem stays of size $<\cfc{\to}$, it can always \emph{escape} the current subsystem (see the remark after Definition~\ref{def:trap:degree});
confluence then lets the two escape paths meet at a common reduct, which becomes the join.
All branching thereby sits at fans, so the resulting cofinal subsystem is globally almost deterministic (Lemma~\ref{lem:regular:S:almost:deterministic}).
In Sections~\ref{sec:regular:subsystem}--\ref{sec:regular:properties}
we construct this cofinal subsystem $\pair{S}{{\substep}}$ of $\pair{A}{{\to}}$.
In Section~\ref{sec:regular:labelling} we show that it admits a decreasing labeling over $\pair{\set{0,1}}{<}$,
which then extends to the entire system by employing one additional label.

\subsection{Constructing the cofinal subsystem}\label{sec:regular:subsystem}

The construction revolves around \emph{fans}: 
points of maximal trap degree $\cfc\to = \trap{A}$ (Figure~\ref{fig:fans}), 
which we use as way pointers.

\begin{figure}[h!]
  \centering
  \includegraphics{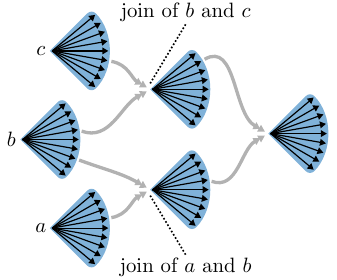}
  \caption{%
    The fans of $a$, $b$ and $c$ in the subsystem $\pair{S}{{\srel}}$.
    The grey lines are deterministic (non-branching).
    The blue points are the fans, each of trap degree $\cfc{\step}$,
    and the grey lines between them are the connecting (escape) paths.
    Each pair of fans has a join which is again a fan.
  }
  \label{fig:fans}
\end{figure}

We start by creating a minimal cofinal subset that only contains elements of trap degree $\cfc{\step}$.
Let $D$ be a cofinal subset of $\pair{A}{{\to}}$ such that $\size{D} = \cfc{\step}$.
Let $C$ be obtained from $D$ by replacing each $a \in D$
by some reduct $b$ of $a$ with the property $\trap{b} = \cfc{\step}$;
these reducts exist since $\pair{A}{{\to}}$ is almost regular.
Then $C$ is a cofinal subset of $\pair{A}{{\to}}$ with $\size{C} = \size{D} = \cfc{\step}$, and
\begin{align}
  \forall c \in C.\; \trap{c} = \size{C} = \cfc{\step} \label{eq:regular:C}
\end{align}
By the axiom of choice, we may assume the existence of a well-ordering of $C$:
\begin{align}
  o: \set{ \alpha \mid \alpha < \size{C} } \to C
\end{align}
So $o$ is a bijection (an enumeration of $C$).

The construction proceeds by a double induction:
the outer induction ensures that every $c \in C$ reaches the subsystem, and
the inner induction is used to construct joins for each pair of fans.

\subsection{Outer induction}

The construction proceeds by induction on ordinals $\alpha < \size{C} = \cfc{\step}$.
We construct
\begin{itemize}
  \item
    a set $F \subseteq C$, whose elements we call \emph{fans}; and
  \item
    a sub-ARS $\pair{S}{{\substep}}$ that contains $F$ together with paths that connect these points.
\end{itemize}
The subsystem $\pair{S}{{\substep}}$ is built up monotonically: 
at each stage we only add to $S$, and at limit stages $S$ is the union of what came before.
We write $F_\alpha$ and $S_\alpha$ for the set $F$ and the subsystem (its carrier together with its steps) at the beginning of stage~$\alpha$,
and $F_{\alpha+1}$ and $S_{\alpha+1}$ for them at the end of stage~$\alpha$;
so $F_0 = S_0 = \emptyset$, and $F_\lambda = \bigcup_{\alpha < \lambda} F_\alpha$ and $S_\lambda = \bigcup_{\alpha < \lambda} S_\alpha$ for limit ordinals~$\lambda$.
Within stage~$\alpha$, $F$ and $\pair{S}{{\substep}}$ denote the current state; thus $F_\alpha \subseteq F \subseteq F_{\alpha+1}$ and $S_\alpha \subseteq S \subseteq S_{\alpha+1}$.

At any stage of the construction, 
we say that nodes $a,b \in F$ \emph{join} in $j \in F$ if
there exist rewrite sequences $a \substep^* j$ and $b \substep^* j$ in the then-current subsystem $\pair{S}{{\substep}}$
such that only the start and end points of these rewrite sequences lie in $F$.
The concept of joins is illustrated in Figure~\ref{fig:fans}.

Throughout, we uphold the following invariants for the state $F_{\alpha+1}$, $S_{\alpha+1}$ at the end of every stage~$\alpha$:
\begin{enumerate}[label=(I\arabic*)]
  \item 
    \emph{size}: $\size{S_{\alpha+1}} \le \max(\size{\alpha},\aleph_0) < \cfc{\step}$;
  \item 
    \emph{determinism}: 
    every node of $S_{\alpha+1} \setminus F_{\alpha+1}$ has at most one $\substep$-successor in $S_{\alpha+1}$; 
  \item 
    \emph{reachability}: 
    every node of $S_{\alpha+1}$ $\substep$-reduces within $S_{\alpha+1}$ to some $a \in F_{\alpha+1}$;
  \item 
    \emph{cofinality}: 
    for every $\beta \le \alpha$, $o(\beta)$ $\step$-reduces to some $a \in F_{\alpha+1}$;
  \item 
    \emph{join}: 
    for all $\beta < \beta' \le \alpha$ with $o(\beta),o(\beta') \in F_{\alpha+1}$,
    $o(\beta)$ and $o(\beta')$ join in some $j \in F_{\alpha+1}$ (in $S_{\alpha+1}$).
\end{enumerate}
We first justify the \emph{size} invariant. This invariant is crucial to ensure that the subsystem $\pair{S}{{\substep}}$ can be escaped.
Each stage~$\alpha$ creates at most one fan plus one join for each $\beta < \alpha$,
and each join consists of two finite rewrite sequences.
Hence each stage~$\alpha$ adds at most $\max(\size{\alpha},\aleph_0)$ nodes.
Taking the union over the $\size{\alpha+1}$ stages up to $\alpha$ gives
\begin{align}
  \size{S_{\alpha+1}} \;\le\; \size{\alpha+1} \cdot \max(\size{\alpha},\aleph_0)
  \;=\; \max(\size{\alpha},\aleph_0)
  \;<\; \cfc{\step}
\end{align}

\subsection{Outer induction step}

We are now at stage $\alpha$ of the induction, and we consider the element $a = o(\alpha)$:
\begin{itemize}
  \item
    If $a \not\in F$ and $a$ reaches $S$, 
    then the \emph{cofinality} invariant already holds for $\beta = \alpha$, 
    and since $a \notin F$ no new joins are required.
    We continue with $\alpha+1$.
  \item
    Otherwise, we add $a$ to $F$ and continue with the construction of joins as described below.
    (If $a$ was already in $F$, then $a$ has been constructed as the result of a prior join.
    We then still need to construct joins for $a$.)
\end{itemize}

\subsection{Inner induction: constructing joins}\label{sec:regular:inner}

We will now proceed with a sub-induction on $\beta < \alpha$ in order to construct pairwise joins for $a$ with elements $o(\beta)$.
This concept of joins is illustrated in Figure~\ref{fig:fans}.
Let $b = o(\beta)$.
If $b \in F$, then we construct a join of $a$ and $b$ as follows.
(If $b \not\in F$, then we continue with $\beta+1$.)

Because $\trap{a} = \trap{b} = \cfc{\step}$ and $\size{S} < \cfc{\step}$ by property~\eqref{eq:regular:C} and the \emph{size} invariant, 
both $a$ and $b$ admit escapes $\zeta_a : a \steps a'$ and $\zeta_b : b \steps b'$, respectively, from the current $S$.
Then the end points $a'$ and $b'$ cannot reach $S$.
By confluence, there exists a common reduct $c \in A$ such that $a' \steps c$ and $b' \steps c$.
By cofinality of $C$, there exists $d \in C$ such that $c \steps d$.
Without loss of generality, we may assume that the escapes all end in $d$, that is,
$\zeta_a : a \steps d$ and $\zeta_b : b \steps d$
(these extensions stay outside of $S$, as the escape endpoints $a'$ and $b'$ can no longer reach it).
We may assume that ($\star$) these paths are simple (by cutting loops)
and that their union is deterministic (drop one branch in case of branching).

Since $d \in C$, there exists some $\delta < \size{C}$ such that $d = o(\delta)$.
We have $\delta > \alpha$ since, by the \emph{cofinality} invariant, every $o(\beta)$ with $\beta \le \alpha$ can reach $S$.
Thus $\delta$ will still be visited in the outer induction (leading to the creation of joins for $d$).

We now add $d$ to $F$ and extend the subsystem $\pair{S}{{\substep}}$ with 
the rewrite sequences $\zeta_a : a \steps d$ and $\zeta_b : b \steps d$.
This establishes the \emph{join} invariant for $a$ and $b$ with witness $d$.
The connecting paths $\zeta_a$ and $\zeta_b$ end in $d \in F$, maintaining \emph{reachability}.
The assumption ($\star$) guarantees that the \emph{determinism} invariant is preserved.
This completes the induction step and, with it, the construction of the
subsystem $\pair{S}{{\substep}}$.

\subsection{Properties of the subsystem}\label{sec:regular:properties}

We collect some important properties of the completed subsystem $\pair{S}{{\substep}}$, where $S = \bigcup_{\alpha} S_\alpha$ and $F = \bigcup_{\alpha} F_\alpha$, the unions over all stages $\alpha < \cfc{\to}$.
Apart from the \emph{size} invariant, which is only needed during the construction, the invariants carry over to the completed subsystem, since each of them concerns only finitely many points, all of which are present from some stage on.
For the \emph{join} invariant we additionally use that no interior point of a join path ever becomes a fan:
every point added to $F$ is new, as it is either some $o(\alpha)$ that does not reach $S$,
or a common reduct $d$ of escape end points, which cannot reach $S$.

\begin{lemma}\label{lem:regular:S:almost:deterministic}
  Let $\pair{A}{{\to}}$ be an almost regular ARS (hence connected and confluent) such that $\cfc{\to}$ is uncountable,
  and let $\pair{S}{{\substep}}$ be the subsystem constructed in Sections~\ref{sec:regular:subsystem}--\ref{sec:regular:inner}.
  Then the subsystem $\pair{S}{{\substep}}$ is globally almost deterministic.
\end{lemma}

\begin{proof}
  Let $a, b \in S$; we join them by two $\substep$-sequences, each with at most one non-deterministic step.
  By the \emph{determinism} invariant, the non-deterministic nodes of $\substep$ are contained in $F$.
  By the \emph{reachability} invariant, $a$ and $b$ admit reductions within $S$ to fans in $F$.
  Let $\sigma_a : a \substep^* a'$ and $\sigma_b : b \substep^* b'$ be reductions of minimal length such that $a', b' \in F$.
  By minimality, these reductions do not pass through a fan from $F$ and thus are deterministic. 
  By the \emph{join} invariant the fans $a'$ and $b'$ join in a common fan $c \in F$,
  that is, $a' \substep^* c$ and $b' \substep^* c$ such that only the start and endpoints of each path lie in $F$.
  Hence each of these two paths starts with a single non-deterministic step
  and then continues deterministically until reaching $c$.
  Concatenating, $a \substep^* a' \substep^* c$ and $b \substep^* b' \substep^* c$
  join $a$ and $b$ at the common reduct $c$, each with at most one non-deterministic step.
\end{proof}

\begin{lemma}\label{lem:regular:S:properties}
  Let $\pair{A}{{\to}}$ be an almost regular ARS (hence connected and confluent) such that $\cfc{\to}$ is uncountable,
  and let $\pair{S}{{\substep}}$ be the subsystem constructed in Sections~\ref{sec:regular:subsystem}--\ref{sec:regular:inner}.
  Then the subsystem $\pair{S}{{\substep}}$ is connected, confluent, and cofinal
  in $\pair{A}{{\to}}$, and $S$ is deterministic away from $F$.
\end{lemma}

\begin{proof}
  We address each of the properties:
  \begin{itemize}
    \item 
      \emph{Cofinality:}
      Every $a \in A$ reduces to some $o(\beta) \in C$ (as $C$ is cofinal), 
      and by the \emph{cofinality} invariant $o(\beta)$ reaches some fan $f$ $\in F \subseteq S$.
    \item 
      \emph{Determinism:} this is precisely the \emph{determinism} invariant.
    \item 
      \emph{Connectedness} and \emph{confluence} follow from Lemma~\ref{lem:regular:S:almost:deterministic}.
      \qedhere
  \end{itemize}
\end{proof}

\subsection{Labeling the subsystem}\label{sec:regular:labelling}

We now prove that $\pair{A}{{\to}}$ is decreasing Church--Rosser over $\pair{\set{0,1,2}}{<}$; \emph{three labels suffice}.
First, we show that the subsystem $\pair{S}{{\substep}}$ can be labeled with $\pair{\set{0,1}}{<}$
as a consequence of the fact that the subsystem is globally (and thus also locally) almost deterministic.

\begin{lemma}\label{lem:almost:deterministic}
  Every locally almost deterministic ARS $\pair{S}{{\substep}}$ is decreasing Church--Rosser over $\pair{\set{0,1}}{<}$.
\end{lemma}

\begin{proof}
  We label all deterministic steps with $0$ and all non-deterministic steps with $1$.
  Consider a peak $b \lsubstep_\alpha a \substep_\beta c$.
  We show that this peak can be closed by a decreasing elementary diagram as shown in Figure~\ref{fig:decreasing:diagram}.
  If $b = c$, we join with empty reductions.
  If $b \ne c$, then both steps are non-deterministic and hence $\alpha = \beta = 1$.
  Since the system is locally almost deterministic, we can join $b$ and $c$ with reductions 
  that each contain at most one non-deterministic step.
  Thus each of the joining reductions contains at most one step with label $1$ and the other steps have label $0$;
  this gives rise to a decreasing elementary diagram.
\end{proof}

By Lemma~\ref{lem:extend}, the labeling of the cofinal subsystem can be extended to the system as a whole
at the cost of one additional label. 
This justifies our first main result:

\begin{theorem}\label{thm:regular}
  Every almost regular ARS $\pair{A}{{\to}}$ (in particular, $\pair{A}{{\to}}$ is connected and confluent)
  has a cofinal, connected and confluent, globally almost deterministic subsystem $\pair{S}{{\substep}}$,
  and is decreasing Church--Rosser over $\pair{\set{0,1,2}}{<}$.
  In particular, this holds for every connected and confluent ARS $\pair{A}{{\to}}$
  whose cofinality cardinal $\cfc{\to}$ is regular.
\end{theorem}

\begin{proof}
  If $\cfc{\to}$ is countable, the claim follows from Lemma~\ref{lem:countable} (the simple cofinal rewrite sequence constructed in its proof is a deterministic subsystem with the required properties).
  Thus assume that $\cfc{\to}$ is uncountable.
  Let $\pair{S}{{\substep}}$ be the subsystem of $\pair{A}{{\to}}$ constructed in Sections~\ref{sec:regular:subsystem}--\ref{sec:regular:inner}.
  By Lemma~\ref{lem:regular:S:almost:deterministic}, $\pair{S}{{\substep}}$ is almost deterministic,
  and consequently decreasing Church--Rosser over $\pair{\set{0,1}}{<}$ by Lemma~\ref{lem:almost:deterministic}. 
  Moreover, $\pair{S}{{\substep}}$ is connected, confluent and cofinal in $\pair{A}{{\to}}$ by Lemma~\ref{lem:regular:S:properties}.
  Hence $\pair{A}{{\to}}$ is decreasing Church--Rosser over $\pair{\set{0,1,2}}{<}$ by Lemma~\ref{lem:extend}.
\end{proof}

The singular case is treated in Section~\ref{sec:fans}.

\section{Completeness for singular cardinals}\label{sec:fans}

\newcommand{\branchpoints}{\mathcal{B}}

The construction of Section~\ref{sec:regular} requires the system to be almost regular:
it relies on a cofinal supply of points of \emph{maximal} trap degree $\cfc{\to}$,
which is guaranteed to exist when $\cfc{\to}$ is regular (Lemma~\ref{lem:trap:degree}).
We now treat the remaining case, in which $\cfc{\to}$ is \emph{singular}.
For singular $\cfc{\to}$ no point needs to have trap degree $\cfc{\to}$ at all~(Example~\ref{ex:aleph:omega}).
We therefore introduce a two-level branching object, shown in Figure~\ref{fig:double:fan},
that we call a \emph{double fan}.
The first-level fan (on the left) is a \emph{hub} from which pairwise disjoint, deterministic escape paths emanate
leading to $\cf{\cfc{\to}}$ second-level fans (on the right) whose trap degree is unboundedly high below $\cfc{\to}$.

\begin{example}[Cofinality $\aleph_\omega$]\label{ex:aleph:omega}
  Fix sets $W_0 \subseteq W_1 \subseteq W_2 \subseteq \cdots$ with $\size{W_i} = \aleph_i$ for every $i \in \nat$.
  Let $A = \set{(n,X) \mid n \in \nat,\ X \subseteq W_n \text{ finite}}$, with steps
  \begin{align*}
    (n,X) \to (n, X \cup \set{\xi})\text{ for }\xi \in W_n \setminus X
    &&\text{and}&&
    (n,X) \to (n+1,X)\;.
  \end{align*}
  The reducts of $(n,X)$ are exactly the elements $(m,Y)$ with $m \ge n$ and $X \subseteq Y$,
  so the system is connected and confluent: $(m,Y)$ and $(m',Y')$ have the common reduct $(\max(m,m'),\, Y \cup Y')$.
  
  We have $\cfc{\to} = \aleph_\omega$:
  if $\size{C} < \aleph_\omega$, then $\size{C} < \aleph_n$ for some $n$,
  so some $\xi \in W_n$ occurs in no second component of an element of $C$,
  and $(n,\set{\xi})$ reaches no element of $C$.
  
  Every $(n,X)$ has trap degree exactly $\aleph_n$:
  its $\aleph_n$ immediate successors form a trap (an escape must start with a step leaving the trap),
  and for $\size{T} < \aleph_n$ a step $(n,X) \to (n, X \cup \set{\xi})$ with a fresh $\xi \in W_n$ escapes $T$.
  Thus the trap degrees are cofinal in $\aleph_\omega$ but no element has maximal trap degree.
  The system is not almost regular, and the fans of Section~\ref{sec:regular} do not exist.
  A double fan instead sends its branches to points of ever higher level~$n$, whose trap degrees $\aleph_n$ are cofinal in $\aleph_\omega$.
\end{example}

The crucial property of the double fan is \emph{survival}.
In the singular case, points of maximal trap degree need not exist (Example~\ref{ex:aleph:omega}),
and any point of smaller trap degree eventually gets trapped as $S$ grows toward~$\cfc{\to}$. 
By spraying out $\cf{\cfc{\to}}$ branches whose trap degrees are cofinal in $\cfc{\to}$, 
a double fan guarantees that at \emph{every} stage some branch is still escapable, keeping the subsystem connectable.

\begin{figure}[h!]
  \centering
  \includegraphics{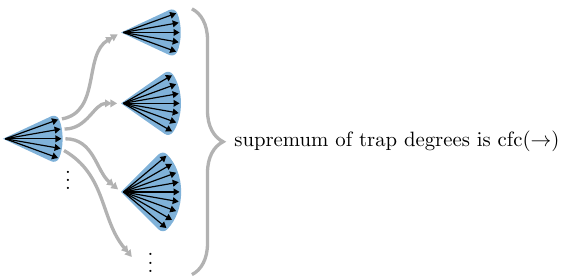}
  \caption{
    Double fan construction. 
    For the second-level fans, displayed on the right, only the source points are part of the double fan,
    their outgoing steps are not.}
  \label{fig:double:fan}
\end{figure}

\paragraph{Outline of the construction: }
The construction proceeds by transfinite induction over an enumeration of a cofinal set $C$, alternating between two procedures.
For each point of $C$ not yet reached by the subsystem, we build a \emph{double fan} (Section~\ref{sec:singular:double:fan}), which establishes reachability,
and then invoke the \emph{join closure} procedure (Section~\ref{sec:singular:closure}).
The join closure joins the branch end points of the new double fan, both \emph{internally} (two branch end points of the same double fan)
and \emph{externally} (with branch end points of earlier double fans).
It is moreover \emph{recursive}, joining all pairs of joins as well.
Every join is formed with a single-fan escape junction of strictly higher trap degree than its two inputs, so that a parent outlives its children.
Because the branches of a double fan have trap degrees cofinal in $\cfc{\to}$, some branch always remains escapable;
together these properties keep the growing subsystem $S$ cofinal, connected, and almost deterministic.

The join closure procedure makes use of commutative binary trees to keep track of what points have already been joined.
The leaves of these trees are the branch end points of the double fans.

\begin{definition}
  We define \emph{(commutative, binary) trees $\trees{L}$ over a set $L$} inductively as follows:
  \begin{enumerate}
    \item $\leaf{l} \in \trees{L}$ for every $l \in L$, and
    \item $\node{l}{r} \in \trees{L}$ for every $l,r \in \trees{L}$.
  \end{enumerate}
\end{definition}

Throughout this section, $\pair{A}{{\to}}$ is a connected and confluent ARS (Convention~\ref{conv:components}).
Then $\trap{A} = \cfc{\to}$. We write $\kappa = \cfc{\to}$.
Assume that $\kappa$ is singular and hence an uncountable limit cardinal.
Let $\lambda = \cf{\kappa}$ together with a function
\begin{align}
  \gamma : \set{ \alpha\text{ an ordinal} \mid \alpha < \lambda} \to \set{\beta\text{ a cardinal} \mid \beta < \kappa}
\end{align}
whose values are cofinal in $\kappa$.
The function $\gamma$ is a \emph{fundamental sequence} for $\kappa$:
$\gamma(\alpha) < \kappa$ for all $\alpha$,
and for every $\beta < \kappa$ there exists some $\alpha$ with $\beta < \gamma(\alpha)$.
Thus the supremum of this sequence is $\kappa$.

Without loss of generality, let $\pair{A}{{\to}}$ be \emph{not} almost regular;
the case of almost regular systems was already covered in Section~\ref{sec:regular}.\footnote{
  In principle, the construction for singular $\cfc{\to}$ can also be applied to almost regular systems.
  However, this requires some extra bookkeeping to remember the intended minimal degree of each node.
  Excluding almost regular systems avoids this bookkeeping since we can always find reducts with larger trap degree.
}
Then there exists an element $a \in A$ such that all reducts $b$ of $a$ have trap degree $\trap{b} < \kappa$.
Let $D$ be a cofinal subset of $\pair{A}{{\to}}$ such that $\size{D} = \kappa$,
and define 
\begin{align}
  C = D \cap \set{b \mid a \steps b}
\end{align}
Then $C$ is again a cofinal subset:
every $c \in A$ must have a common reduct with $a$ by confluence and connectedness,
and this common reduct rewrites to some $d \in D$;
then $d$ is also a reduct of $a$ and thus in $C$.
Thus $C$ is a cofinal subset and $\size{C} \le \size{D} = \cfc{\step}$, hence $\size{C} = \cfc{\step} = \kappa$.
Moreover
\begin{align}
  \forall c \in C.\; \forall d \in A.\; c \steps d \implies \trap{d} < \kappa \;, \label{eq:singular:C}
\end{align}
by construction of $C$, that is,
no element of $C$ has a reduct with the maximal trap degree $\kappa$.

\subsection{Escape junctions}

Throughout this construction,
we repeatedly need to escape a set $S \subseteq A$ of size $< \kappa$ from finitely many points $E$,
and join the escape paths in a common reduct with a trap degree \emph{strictly} above some cardinal $\beta < \kappa$.
We summarize this construction in the following lemma.

\begin{lemma}[Escape junctions]\label{lem:escape:junction}
  Let the following be given:
  \begin{enumerate}
    \item $\pair{A}{{\to}}$ a confluent, connected ARS such that $\cfc{\step}$ is uncountable and singular;
    \item $S \subseteq A$ a set with $\size{S} < \cfc{\step}$;
    \item $E \subseteq A$ a non-empty finite set of points such that each $e \in E$ has an escape from $S$;
    \item $\beta < \cfc{\step}$ a cardinal;
  \end{enumerate}
  There exists a point $\junction{S}{E}{\beta} \in A$ and, for each $e \in E$, a rewrite sequence
  $\pi_e : e \steps \junction{S}{E}{\beta}$, such that:
  \begin{enumerate}
    \item $\trap{\junction{S}{E}{\beta}} > \beta$;
    \item $\junction{S}{E}{\beta}$ cannot reach $S$, that is, $\junction{S}{E}{\beta} \not\to^* s$ for every $s \in S$;
    \item each $\pi_e$ is simple and meets $S$ only in its starting point; and
    \item the union $\bigcup_{e \in E} \pi_e$ is deterministic (every node has at most one successor).
  \end{enumerate}
\end{lemma}

\begin{proof}
  By assumption, each $e \in E$ admits an escape $\zeta_e : e \steps e'$ from $S$; in particular,
  $e'$ cannot reach $S$.
  By connectedness and confluence, there exists a common reduct $f \in A$ such that $e' \steps f$ for every $e \in E$.
  By Lemma~\ref{lem:trap:degree}, there exists $g \in A$ such that $f \steps g$ and $\trap{g} \ge \beta^+$.
  Here we use that $\cfc{\step}$ is singular and hence a limit cardinal;
  thus $\beta < \cfc{\step}$ implies $\beta^+ < \cfc{\step}$.
  Without loss of generality, we may assume that the escapes all end in $g$, that is,
  $\zeta_e : e \steps g$ for every $e \in E$
  (these extensions stay outside of $S$, as the escape endpoints can no longer reach it).
  Moreover, we may assume that these paths are simple (by cutting loops)
  and that their union is deterministic (drop one branch in case of branching).
  Then $g$ has all the required properties.
\end{proof}

\subsection{Induction}\label{sec:singular:induction}

Recall that $C$ is a cofinal subset of $\pair{A}{{\to}}$ with property~\eqref{eq:singular:C}. Let
\begin{align}
  o: \set{ \alpha \mid \alpha < \size{C} } \to C
\end{align}
be a bijection (an enumeration of $C$).

The construction proceeds by induction on ordinals $\alpha < \size{C}$.
We construct
\begin{itemize}
  \item
    a set $H \subseteq A$ (consisting of hubs of double fans constructed so far);
  \item
    a set $F \subseteq A$ (consisting of branch end points of double fans constructed so far);
  \item
    a partial function $\sjoins: \trees{A} \pmap A$; and
  \item
    a sub-ARS $\pair{S}{{\substep}}$ that contains $F$ together with paths that connect these points.
\end{itemize}
These ingredients are built up monotonically:
at each stage we only extend $H$, $F$, $\sjoins$ and $\pair{S}{{\substep}}$,
and at limit stages we take the union.
Initially $H = \emptyset$, $F = \emptyset$, $\dom{\sjoins} = \emptyset$, and $S = \emptyset$.
We write $H_\alpha$, $F_\alpha$, $\sjoins_\alpha$ and $S_\alpha$ for these ingredients at the beginning of stage~$\alpha$ 
(where $S_\alpha$ comprises the carrier and the steps of the subsystem),
and $H_{\alpha+1}$, $F_{\alpha+1}$, $\sjoins_{\alpha+1}$ and $S_{\alpha+1}$ for them at the end of stage~$\alpha$;
at limit stages we take unions.
Within stage~$\alpha$, $H$, $F$, $\sjoins$ and $\pair{S}{{\substep}}$ denote the current state,
which lies between the states at the beginning and at the end of the stage; for instance, $S_\alpha \subseteq S \subseteq S_{\alpha+1}$.
As $\sjoins$ is only ever extended, we write $\joins{t}$ regardless of the stage.

\begin{notation}
  We write $a \adet b$ if there exists an almost deterministic rewrite sequence from $a$ to $b$ in the subsystem $\pair{S}{{\substep}}$.
\end{notation}

For stage $\alpha$, we define a \emph{size threshold}
\begin{align}
  \mu_\alpha \;=\; \max(\size{\alpha},\; \lambda,\; \omega) \;<\; \kappa
\end{align}
We say that a point $a \in A$ is \emph{alive (at stage $\alpha$)} if $\trap{a} > \mu_\alpha$.
Since $\size{S} \le \mu_\alpha$ throughout stage~$\alpha$, an alive point can still escape the current subsystem, and hence is still usable to form new joins;
the bookkeeping below serves precisely to keep the relevant points alive long enough.
The bound $\mu_\alpha$ is chosen in such a way that a point that is alive at the beginning of the stage will still be alive at the end of the stage.
This is necessary for the join closure construction (Section~\ref{sec:singular:closure})
as we must prevent a point that participates in the closure construction (was alive at the start)
from getting trapped during the construction (and thus is no longer able to escape the growing subsystem).

The invariants refer to the \emph{hubs} (the set $H$), from which the branches of a double fan emanate,
the \emph{branch end points} of double fans (the set $F$), and the \emph{join points} $\joins{t}$ for $t \in \dom{\sjoins}$;
these are constructed in Sections~\ref{sec:singular:double:fan} and~\ref{sec:singular:closure}.
Throughout, we uphold the following invariants for the state $H_{\alpha+1}$, $F_{\alpha+1}$, $\sjoins_{\alpha+1}$, $S_{\alpha+1}$ at the end of every stage~$\alpha$, which fall into four groups by purpose.
The \emph{size}, \emph{sub-maximal degree}, and \emph{survival} invariants keep the subsystem small enough, and its points high enough in trap degree, to remain escapable.
The \emph{cofinality} invariant ensures that every point is eventually reached.
The \emph{join structure}, \emph{join closure}, and \emph{degree monotonicity} invariants track which points have been joined, and keep every join alive at least as long as its inputs.
Finally, the \emph{determinism} invariant confines all branching to hubs and join points.
In detail:
\begin{enumerate}[label=(J\arabic*)]
  \item \label{it:singular:size}
    \emph{size}: $\size{F_{\alpha+1}} \le \mu_\alpha$ and $\size{S_{\alpha+1}} \le \mu_\alpha$;
  \item \label{it:singular:sub:max}
    \emph{sub-maximal degree}: $\trap{a} < \kappa$ for every $a \in S_{\alpha+1}$;
  \item \label{it:singular:survival}
    \emph{survival}:
    every hub, branch end point and join point created at stage $\alpha$ is alive at $\alpha$
    (the interior nodes of the added escape paths need not be alive; they only need to be fresh, that is, not yet in $S$);
  \item \label{it:singular:cofinality}
    \emph{cofinality}: every $o(\beta)$, $\beta\le\alpha$, $\to$-reduces to some point in $S_{\alpha+1}$;
  \item \label{it:singular:join}
    \emph{join structure} (where $\substep$ and $\adet$ refer to $S_{\alpha+1}$):
    \begin{itemize}
      \item $\dom{\sjoins_{\alpha+1}}\subseteq\trees{F_{\alpha+1}}$,
      \item $\joins{\leaf{a}} = a$ for all $a \in F_{\alpha+1}$,
      \item if $\joins{t_1} \substep x_1 \substep \cdots \substep x_n = \joins{t}$ with $n \ge 1$ and none of $x_1,\ldots,x_{n-1}$ is of the form $\joins{s}$,
        then $t = \node{t_1}{t_2}$ for some $t_2$, and
      \item for all $t = \node{t_1}{t_2}\in\dom{\sjoins_{\alpha+1}}$, we have $\joins{t_1} \adet \joins{t} \ladet \joins{t_2}$;
    \end{itemize}
  \item \label{it:singular:closure}
    \emph{join closure}:
    whenever $\node{t_1}{t_2} \in \trees{F_{\alpha+1}}$ has both $\joins{t_1},\joins{t_2}$ alive at stage~$\alpha$,
    $\node{t_1}{t_2} \in \dom{\sjoins_{\alpha+1}}$;
  \item \label{it:singular:monotonicity}
    \emph{degree monotonicity}: if $t \in \dom{\sjoins_{\alpha+1}}$, then
    \begin{align}
      \trap{\joins{t}} > \trap{\joins{s}} \quad \text{for every strict subtree $s$ of $t$}
    \end{align}
    So a join of alive points is alive;
  \item \label{it:singular:determinism}
    \emph{determinism}:
    with $\branchpoints'_{\alpha+1} = \set{\joins{t} \mid t \in \dom{\sjoins_{\alpha+1}}}$
    and $\branchpoints_{\alpha+1} = H_{\alpha+1} \cup \branchpoints'_{\alpha+1}$,
    every node of $S_{\alpha+1} \setminus \branchpoints_{\alpha+1}$ is $\substep$-deterministic in $S_{\alpha+1}$,
    and $\substep$-reduces within $S_{\alpha+1}$ to a point in $\branchpoints'_{\alpha+1}$ (via deterministic steps).
\end{enumerate}
We immediately state the justifications for the invariants:
\begin{itemize}
  \item
    Invariant~\ref{it:singular:size} holds since in every stage $\alpha$:
    \begin{itemize}
      \item
        We add at most one double fan (consisting of $\lambda$ finite rewrite sequences),
        and extend $F$ by at most $\lambda$ points (the branch endings of the double fan).
      \item
        We construct at most $\size{\trees{F}} \le \mu_\alpha$\footnote{
          The set $\trees{F}$ has $\le \mu_\alpha$ leaves, and $\le \mu_\alpha \cdot \mu_\alpha = \mu_\alpha$ trees of height $n \in \nat$.
          Thus $\size{\trees{F}} \le \mu_\alpha \cdot \aleph_0 = \mu_\alpha$.
        } joins (each consisting of two finite rewrite sequences).
    \end{itemize}
    Thus the $\le \size{\alpha} \le \mu_\alpha$ earlier stages,
    each contribute $\le \mu_\alpha$, and $\mu_\alpha \cdot \mu_\alpha = \mu_\alpha$.
  \item
    Invariant~\ref{it:singular:sub:max} follows from Equation~\eqref{eq:singular:C} since every $a \in S$ is a reduct of some $c \in C$.
  \item
    Invariant~\ref{it:singular:survival} follows since the hub and branch degrees of a new double fan are taken $> \mu_\alpha$,
    and a join carries $\trap{j_{a,b}} > \max(\trap{a},\trap{b})$, strictly above both of its inputs;
    the interior nodes of the added paths are fresh by Lemma~\ref{lem:escape:junction}, since each path meets $S$ only in its starting point.
  \item
    Invariant~\ref{it:singular:cofinality} follows since, at each stage $\alpha$, we check whether $o(\alpha)$ reaches $S$;
    if it does not, we construct a double fan for $o(\alpha)$ to establish reachability.
  \item
    Invariant~\ref{it:singular:join} is established in the construction in Section~\ref{sec:singular:closure}.
  \item
    Invariant~\ref{it:singular:closure}:
    suppose $\joins{t_1}$ and $\joins{t_2}$ are alive at stage $\alpha$.
    Since $F$ grows only at double fan creation stages and aliveness only decreases,
    both were already alive at the last fan creation stage $\beta \le \alpha$ with $\node{t_1}{t_2} \in \trees{F}$.
    Its closure reaches $\node{t_1}{t_2}$ after $t_1$ and $t_2$;
    by the size invariant $\size{S} \le \mu_\beta < \trap{\joins{t_i}}$ throughout,
    so both still escape $S$ and the join is formed, unless $\node{t_1}{t_2}$ was already joined at an earlier stage.
    Either way $\node{t_1}{t_2} \in \dom{\sjoins}$.
  \item
    Invariant~\ref{it:singular:monotonicity}:
    for the children $s$ of $t$ the invariant follows immediately since
    $\joins{t}$ is chosen with this property,
    and for other subtrees it follows by induction on the tree structure.
  \item
    Invariant~\ref{it:singular:determinism} follows since, by construction,
    the only branching points of $S$ are the hubs~$H$ of double fans,
    and points participating in the join closure;
    the latter are precisely the points in the image of $\sjoins$ and they contain the branch endings of double fans.
    Moreover, by construction all paths added to $\pair{S}{{\substep}}$ are escape paths ending in a node in $\branchpoints'$.
\end{itemize}

\subsection{Induction step}\label{sec:singular:step}

We are now at stage $\alpha$ of the induction, and we consider the element $a = o(\alpha)$:
\begin{itemize}
  \item
    If $a$ reaches $S$,
    then the \emph{cofinality} invariant already holds for $\beta = \alpha$, and no new joins are required.
    We continue with $\alpha+1$.
  \item
    Otherwise, $a$ does not reach $S$ and hence can escape $S$ by Lemma~\ref{lem:escape:or:reach};
    we construct a \emph{double fan} for $a$ (see below).
    Afterwards, we extend
    \begin{itemize}
      \item $H$ with $\hub{a}$,
      \item $F$ with all branch end points $\set{ \branch{a}{\beta} \mid \beta < \lambda }$,
    \end{itemize}
    and then we execute the join closure procedure (described below).
\end{itemize}

\subsection{Double fan}\label{sec:singular:double:fan}

We now describe the construction of the \emph{double fan}.
By the design of the induction step (Section~\ref{sec:singular:step}),
this procedure is invoked only for points $a \in A$ that can escape the current subsystem $S$
(namely, points that do not reach $S$; see Lemma~\ref{lem:escape:or:reach}).

The \emph{double fan for $a$} (where $a$ can escape $S$) is constructed as follows:
\begin{enumerate}
  \item
    The \emph{hub} $\hub{a} = \junction{S}{\set{a}}{\mu_\alpha}$.
    So $a \steps \hub{a}$, $\trap{\hub{a}} > \mu_\alpha$, and $\hub{a}$ cannot reach $S$.
    The hub is well-defined since $a$ can escape $S$ by assumption.
  \item
    We first add $\hub{a}$ to $S$ as an isolated point (it is new, since it cannot reach $S$).
    The \emph{branches} of the double fan are then constructed by induction on $\beta < \lambda$:
    \begin{align*}
      \branch{a}{\beta} \;=\; \junction{S}{\set{\hub{a}}}{\max(\gamma(\beta),\mu_\alpha)} \;,
    \end{align*}
    and, in each step, we extend $\pair{S}{{\substep}}$ with the path $\pi_{\hub{a}}$ from the escape junction construction.
    We argue that the construction never gets stuck.
    At each step $\beta$, we only add finitely many points to $S$.
    Hence throughout the branch construction we have
    \begin{align}
      \size{S} \le \mu_\alpha + \lambda \cdot \aleph_0 = \mu_\alpha < \trap{\hub{a}} \label{eq:branch:construction:S}
    \end{align}
    Thus, $\hub{a}$ can escape $S$ at every step.
\end{enumerate}
After the completion of the double fan we have $\size S \le \mu_\alpha$ by~\eqref{eq:branch:construction:S},
and
\begin{align}
  \trap{\branch a\beta} > \max(\gamma(\beta),\mu_\alpha) \ge \size S
\end{align}
So there are no stillborn branches.

The double fan contains peaks \emph{at} its hub and these peaks must be joined (in a decreasing way).
For this purpose, the induction step (Section~\ref{sec:singular:step}) will invoke the join closure procedure (Section~\ref{sec:singular:closure}) 
to construct almost deterministic rewrite sequences joining these peaks.
We need to construct these joins \emph{eagerly}, directly after the construction of the double fan.
At a later point in the construction,
some of the branches might no longer be able to escape the growing subsystem.
The meeting points themselves must be joined as well, and so on.
For this purpose, in the induction step (Section~\ref{sec:singular:step}), after the construction of the double fan,
the set $F$ is extended with all the branch end points and
then the \emph{join closure} is invoked to recursively join all points in $F$.

\subsection{Join closure}\label{sec:singular:closure}

We are now describing the \emph{join closure} procedure.
This procedure is invoked whenever the subsystem is extended with a double fan.
The set $F$ contains all points that need to be joined, and, in particular,
it contains the branch end points of all double fans that have been constructed.

The join closure is recursive, we also need to join all pairs of joins, and so on.
To keep track of what we have already joined we use the partial function
\begin{align}
  \sjoins: \trees{A} \pmap A
\end{align}
At each stage of the induction, the domain of this function is a subset of $\trees{F}$.
The leaves of these trees are branch end points of double fans that we have constructed,
and the inner nodes represent joins of their children.

Throughout the construction, we will uphold the following invariants:
\begin{align}
  \node{t_1}{t_2} \in \dom{\sjoins}
  \implies \joins{t_1} \adet \joins{\node{t_1}{t_2}} \ladet \joins{t_2}
\end{align}
Put in words: the points represented by $t_1$ and $t_2$ admit almost deterministic rewrite sequences to
the point represented by $\node{t_1}{t_2}$.

The \emph{join closure} procedure is executed as follows.
We process the set $\trees{F}$ in a well-founded order that respects the size of the trees
(trees with fewer nodes come first).
For each $t \in \trees{F}$:
\begin{enumerate}
  \item
    If $t = \leaf{a}$ for $a \in F$ and $t \not\in \dom{\sjoins}$, then extend $\sjoins$ with $t \mapsto a$.
  \item
    If $t = \node{t_1}{t_2}$, $t_1, t_2 \in \dom{\sjoins}$, and $t \not\in \dom{\sjoins}$,
    then let
    \begin{align*}
      a = \joins{t_1} && b = \joins{t_2}
    \end{align*}
    If $a$ and $b$ are alive, then construct
    \begin{align}
      j_{a,b} = \junction{S}{\set{a,b}}{\max(\trap{a}, \trap{b})} \;, \label{eq:tower:join}
    \end{align}
    Here Lemma~\ref{lem:escape:junction} is applicable since
    \begin{itemize}
      \item $a$ and $b$ are alive and hence $\trap{a},\trap{b} > \mu_\alpha \ge \size{S}$, and
      \item $\max(\trap{a},\trap{b}) < \kappa$ by invariant~\ref{it:singular:sub:max}.
    \end{itemize}
    We extend $\pair{S}{{\substep}}$ with the paths $\pi_a$ and $\pi_b$ from the escape junction construction (Lemma~\ref{lem:escape:junction}),
    and extend $\sjoins$ with $t \mapsto j_{a,b}$.
    Observe that, by construction, we have $\trap{j_{a,b}} > \max(\trap{a},\trap{b})$, justifying invariant~\ref{it:singular:monotonicity}.
\end{enumerate}
It is important to note that the joins are constructed in such a way that trap degree of the \emph{parent} $\joins{\node{t_1}{t_2}}$
is strictly greater than the trap degrees of its \emph{children} $\joins{t_1}$ and $\joins{t_2}$.
This ensures that the parent is alive at least as long as its children.
As a consequence, whenever a child participates in another join, creating a new peak at the child,
the parent is still alive and the peak can immediately be closed at the parent.

The joins themselves are constructed as single fans, not double fans.
The double fans serve as the basis of the construction, and their branch endings are then joined using single fans.
The purpose of the double fans is the following.
Throughout the construction, every single point will eventually die (will no longer be alive).
This is because the subsystem $\pair{S}{{\substep}}$ grows towards size $\cfc{\step}$ 
while every point $d$ in the reduction graph of $C$ has trap degree $\trap{d}$ $< \cfc{\step}$.
Thus every single fan, every hub, and every branch of a double fan will eventually cease to participate in the join closure.
However, as the trap degrees of the branches of a double fan form a cofinal sequence with supremum $\cfc{\step}$, 
at every stage of the construction some of the branches survive.
These branches have trap degrees arbitrarily high $< \cfc{\step}$;
their participation in joins ensures that every point (participating in joins) has parents with 
trap degree arbitrarily high $< \cfc{\step}$.
This is important such that some of these parents keep participating in joins, to keep the entire subsystem connected.

\subsection{Properties of the subsystem}\label{sec:singular:properties}

From now on, $H$, $F$, $\sjoins$ and $\pair{S}{{\substep}}$ denote the completed construction, that is, the unions over all stages $\alpha < \kappa$.
By invariant~\ref{it:singular:determinism}, branching in $\pair{S}{{\substep}}$ occurs only at the points
\begin{align*}
  \branchpoints \;=&\; H \cup \branchpoints' & \branchpoints' &= \set{ \joins{t} \mid t \in \dom{\sjoins} } ,
\end{align*}
Apart from the \emph{size}, \emph{survival} and \emph{join closure} invariants, which are only used for specific stages, the invariants carry over to the completed construction, since each of them concerns only finitely many points, all of which are present from some stage on.
For the \emph{determinism} invariant and the almost deterministic sequences in the \emph{join structure} invariant, we additionally use that only points of $\branchpoints$ ever receive new steps:
every path added to $S$ is a path $\pi_e$ of Lemma~\ref{lem:escape:junction}, which meets $S$ only in its starting point,
and this starting point is a hub or a point $\joins{t}$.
Moreover, every hub, branch end point and join point is new when it is created (it cannot reach the then-current $S$),
so a point of $S_{\alpha+1} \setminus \branchpoints_{\alpha+1}$ never enters $\branchpoints$ later.

\begin{lemma}\label{lem:singular:subsystem}
  Let $\pair{A}{{\to}}$ be a connected and confluent ARS that is not almost regular and such that $\cfc{\to}$ is singular,
  and let $\pair{S}{{\substep}}$ be the subsystem constructed in Sections~\ref{sec:singular:induction}--\ref{sec:singular:closure}.
  Then the subsystem $\pair S \substep$ is cofinal in $\pair A \to$ and deterministic away from $\branchpoints$,
  and every point in $S \setminus \branchpoints$ $\substep$-reduces to a point in $\branchpoints'$.
  Moreover, any two points of $S$ have a common $\substep$-reduct in $S$; in particular, $\pair{S}{{\substep}}$ is connected and confluent.
\end{lemma}

\begin{proof}
  The cofinality follows since every $a \in A$ reduces to some $o(\alpha) \in C$,
  and $o(\alpha)$ reaches $S$ by invariant~\ref{it:singular:cofinality}.
  The determinism and reachability are precisely invariant~\ref{it:singular:determinism}.
  For the last claim, it suffices to join two points $\joins{t_x}$ and $\joins{t_y}$,
  since every point of $S$ reduces to such a point (a hub to its branch end points).
  Let $\alpha$ be a stage by which $t_x, t_y \in \dom{\sjoins}$.
  As $\gamma$ is cofinal in $\kappa$, the double fan of stage~$0$ has a branch end point $f$ with $\trap{f} > \mu_\alpha$,
  so $f$ is alive at all stages $\le \alpha$.
  By the survival, join closure and degree monotonicity invariants, $\joins{t_x}$ is joined with $f$
  (at the stage where $t_x$ is added) into a point of trap degree $> \mu_\alpha$, and likewise $\joins{t_y}$.
  These two points are alive at stage~$\alpha$, and are hence joined as well.
\end{proof}

\begin{lemma}\label{lem:singular:S:almost:deterministic}
  Let $\pair{A}{{\to}}$ be a connected and confluent ARS that is not almost regular and such that $\cfc{\to}$ is singular,
  and let $\pair{S}{{\substep}}$ be the subsystem constructed in Sections~\ref{sec:singular:induction}--\ref{sec:singular:closure}.
  Then the subsystem $\pair{S}{{\substep}}$ is locally almost deterministic.
\end{lemma}

\begin{proof}
  Let $a,b,c \in S$ such that $b \lsubstep a \substep c$ and $b \neq c$.
  By Lemma~\ref{lem:singular:subsystem}, there exist deterministic reductions
  $b \substep^* b'$ and $c \substep^* c'$ to some $b',c' \in \branchpoints'$ (zero steps if $b \in \branchpoints'$, respectively $c \in \branchpoints'$).
  Here $b, c \notin H$, since no step of $S$ ends in a hub: a hub is added to $S$ as an isolated point, and every path added afterwards is simple and meets the then-current $S$ only in its starting point.
  We choose these reductions of minimal length, so that $b'$ and $c'$ are the first points of $\branchpoints'$ on $a \substep b \substep^* b'$ and $a \substep c \substep^* c'$ after~$a$.
  By definition of $\branchpoints'$, we have $b' = \joins{t_b}$ and $c' = \joins{t_c}$ for some $t_b, t_c \in \dom{\sjoins}$.
  We show that $b'$ and $c'$ were alive at a common stage.
  Since $a$ is not deterministic, $a \in \branchpoints$ by invariant~\ref{it:singular:determinism}.
  Then $a$ is either a hub $H$ or a join point $\branchpoints'$:
  \begin{itemize}
    \item
      If $a \in H$,
      then $b$ and $c$ start two distinct branches of a double fan,
      and $b'$, $c'$ must be the branch end points of that fan: $b',c' \in F$.
      Both are created at the stage of this fan creation,
      and hence alive throughout this stage by invariant~\ref{it:singular:survival}.
    \item
      If $a = \joins{t_a}$ for some $t_a \in \dom{\sjoins}$,
      then
      \begin{align*}
        b' &= \joins{\node{t_a}{s_1}} & c' &= \joins{\node{t_a}{s_2}}
      \end{align*}
      for some $s_1, s_2$ by the third item of invariant~\ref{it:singular:join}, applied to the reductions $a \substep b \substep^* b'$ and $a \substep c \substep^* c'$.
      Each of these joins was formed while $a$ was alive (the closure joins only alive points).
      At the later of the two stages, say stage $\alpha$,
      $a$ is alive, and by degree monotonicity~\ref{it:singular:monotonicity}
      $\trap{b'},\trap{c'} > \trap{a} > \mu_\alpha$.
      Hence $b'$ and $c'$ are alive at stage~$\alpha$.
  \end{itemize}
  In either case $b' = \joins{t_b}$ and $c' = \joins{t_c}$ are alive at a common stage.
  So by the join closure invariant~\ref{it:singular:closure} $\node{t_b}{t_c} \in \dom{\sjoins}$,
  and by the join structure invariant~\ref{it:singular:join}
  \begin{align*}
    b' = \joins{t_b} \adet d \ladet \joins{t_c} = c'
    &&\text{where}&&
    d = \joins{\node{t_b}{t_c}}
  \end{align*}
  Composing the deterministic sequences $b \substep^* b'$ and $c \substep^* c'$
  with the almost deterministic $b' \adet d$ and $c' \adet d$ yields almost
  deterministic rewrite sequences $b \substep^* d$ and $c \substep^* d$.
  Hence the peak $b \lsubstep a \substep c$ is joined by almost deterministic
  rewrite sequences.
\end{proof}

\section{Results}\label{sec:singular:results}

In this section we assemble our completeness results from the constructions of the previous sections.
Since a labeling can be built independently on each connected component, it suffices to treat a
single connected, confluent ARS. The regular case was settled in Section~\ref{sec:regular}
(Theorem~\ref{thm:regular}); the next theorem disposes of the remaining singular case, via the
double-fan construction of Section~\ref{sec:fans}.

\begin{theorem}\label{thm:singular:almost}
  Every connected, confluent ARS $\pair{A}{{\to}}$ whose cofinality cardinal $\cfc{\to}$ is singular
  has a cofinal, connected and confluent, locally almost deterministic subsystem $\pair{S}{{\substep}}$,
  and is decreasing Church--Rosser over \mbox{$\pair{\set{0,1,2}}{<}$}.
\end{theorem}

\begin{proof}
  The case of almost regular $\pair{A}{{\to}}$ is covered by Theorem~\ref{thm:regular}.
  Let $\pair{S}{{\substep}}$ be the subsystem of $\pair{A}{{\to}}$ constructed in Sections~\ref{sec:singular:induction}--\ref{sec:singular:closure}.
  By Lemma~\ref{lem:singular:S:almost:deterministic}, $\pair{S}{{\substep}}$ is locally almost deterministic,
  and thus decreasing Church--Rosser over $\pair{\set{0,1}}{<}$ by Lemma~\ref{lem:almost:deterministic}.
  Moreover, $\pair{S}{{\substep}}$ is cofinal in $\pair{A}{{\to}}$ by Lemma~\ref{lem:singular:subsystem}, and connected and confluent by the same lemma.
  Hence $\pair{A}{{\to}}$ is decreasing Church--Rosser over $\pair{\set{0,1,2}}{<}$ by Lemma~\ref{lem:extend}.
\end{proof}

With both cases settled, we obtain our main result:
the decreasing diagrams technique is complete, with the optimal three labels and no cardinality restriction.
This resolves Problem~56 of the RTA List of Open Problems~\cite{rtaloop56}.

\begin{theorem}[Main theorem]\label{thm:main}
  Every confluent ARS $\pair{A}{{\to}}$ is decreasing Church--Rosser over $\pair{\set{0,1,2}}{<}$.
\end{theorem}

\begin{proof}
  It suffices to consider each connected component separately.
  For components with countable $\cfc{\to}$, this is Lemma~\ref{lem:countable}.
  For uncountable almost regular components, the claim follows from Theorem~\ref{thm:regular}.
  The remaining singular cases are covered by Theorem~\ref{thm:singular:almost}.
\end{proof}

\begin{corollary}\label{cor:cr:dcr} $\CR = \mathrm{DCR}$: the Church--Rosser and decreasing Church--Rosser properties are equivalent.
\end{corollary}
\begin{proof}
  Immediate from Theorem~\ref{thm:main} and van~Oostrom's Theorem~\ref{thm:decreasing:diagrams}.
\end{proof}
\noindent
This is arguably the most compact form of our result: decreasing diagrams forfeit no power as a confluence criterion, since every confluent system is provably confluent by this single method.

At the cardinality of greatest practical interest --- and without any set-theoretic assumption --- we obtain:
\begin{corollary}\label{cor:dcr:complete:reals} The decreasing diagrams method is complete for proving confluence of ARS on real numbers (regardless of whether the Continuum Hypothesis is assumed).
\end{corollary}

As a consequence of our main theorem,
van~Oostrom's decreasing diagrams~\cite{oost:1994b,oost:1994} as well as 
de~Bruijn's weak diamond property~\cite{brui:1978} are \emph{complete} for confluence
(the equivalence is due to Endrullis and Klop~\cite{endr:klop:2013}).
The constructions and proofs of this paper have been formalized in Isabelle/HOL as well as in Lean
(both formalizations are archived on Zenodo~\cite{endr:ivan:raam:2026:zenodo}).

\section{Conclusion and future work}\label{sec:conclusion}

We have proved that van~Oostrom's decreasing diagrams method with three labels is complete for proving confluence, 
without any cardinality restriction, and we have obtained two independent machine-checked formalizations (in Lean and Isabelle/HOL) that confirm this.

Termination and confluence are the two central correctness properties of programs and rewrite systems: 
the former guarantees that every computation eventually halts, the latter that its outcome is unique. 
For termination, criteria that are complete in principle have long been available: 
a term rewriting system terminates if and only if it is compatible with a well-founded monotone algebra~\cite{zant:1994},
and the dependency-pair framework~\cite{arts:gies:2000,gies:thie:schn:falk:2006} provides a complete characterization as well.
For confluence, by contrast, no comparably complete method was known. Our result closes this gap: confluence is now on the same footing as termination, certified in principle by a single, local criterion.

At its core, the proof establishes the existence of a (globally or locally) almost deterministic, connected, cofinal sub-ARS 
--- equivalently, a rewrite strategy --- in every confluent ARS of arbitrarily high cardinality, 
and uses this sub-ARS to construct a locally decreasing labeling of the rewrite steps.
That such a proof can be carried out in ZFC without additional axioms is perhaps surprising, for two reasons:
\begin{itemize}
  \item
    in 1994 van~Oostrom formulated a conjecture~\cite[Conjecture 2.3.31]{oost:1994} on the incompleteness of decreasing diagrams; 
    our result shows that this conjecture cannot hold in ZFC;
  \item
    confluent ARS are closely related to directed sets, 
    yet it is known~\cite{todo:1985} that (barring inconsistency) ZFC 
    does not fix the number of cofinal types of directed sets already at the smallest uncountable cardinality --- 
    so one might have expected the cofinal behavior of confluent ARS to be undetermined in ZFC as well.
\end{itemize}

The generality of our treatment pays off concretely for the case of greatest practical interest: 
because completeness of decreasing diagrams holds at every cardinality, it holds in particular for systems involving the \emph{real numbers}, 
and it does so unconditionally --- independent of the Continuum Hypothesis (Corollary~\ref{cor:dcr:complete:reals}).

We propose the following directions for future work:
\begin{enumerate}
  \item
    investigate the dependency of the completeness proof on the Axiom of Choice;
  \item
    investigate under what conditions a decreasing labeling is \emph{computable}, or exists in a constructive sense
    (our proof is non-constructive and relies on the Axiom of Choice; we claim no effective content for the labelings it yields);
    for restricted classes such as left-linear term rewriting systems, decreasing labelings can already be
    constructed automatically via rule labeling and relative termination~\cite{hiro:midd:2011,aoto:2010,nage:felg:zank:2017};
  \item
    complementing the previous point, develop fruitful \emph{heuristics} for applying decreasing diagrams
    (and the weak diamond property) to \emph{real-world programs}, extending both the confluence analysis of
    actor programs~\cite{henrio:etal:2026} and the rule-labeling techniques and tools for term rewriting
    systems~\cite{hiro:midd:2011,aoto:2010,nage:felg:zank:2017,zank:felg:midd:2011:tool};
    since our result guarantees that a locally decreasing labeling always exists, the central practical
    challenge becomes finding a good one;
  \item
    identify classes of (hybrid) models for which confluence guarantees \emph{stability}, as it does for the bouncing ball of Remark~\ref{rem:uncountable:ars};
    confluence is related to, but coarser than, stability, since it ignores the topology on the states;
  \item
    develop \emph{modular} techniques for constructing decreasing labelings ---
    building a locally decreasing labeling of a system compositionally from labelings of its parts;
    the decreasing-diagrams proof of the modularity of confluence~\cite{klop:oost:vrij:2000} is one such instance,
    and a general compositional theory would be valuable;
  \item
    generalize results in rewriting theory that rely on Klop's cofinality property
    (e.g., modularity of confluence~\cite{toya:1987}, proved via cofinal derivations in~\cite{jou:liu:2012})
    to the setting of almost deterministic strategies; see~\cite{ivanov:2026a} for related concepts;
  \item
    explore the practical implications of completeness for ARS on the real numbers, 
    in particular for formal methods for hybrid systems and cyber-physical systems~\cite{plat:2010};
  \item
    unlike for confluence, van~Oostrom's extension of decreasing diagrams to commutation~\cite{oost:1994} 
    is \emph{not} complete for proving commutation (see~\cite{endr:klop:2013,endr:klop:over:2020}); 
    it would be interesting to find a generalization that is complete for proving commutation.
\end{enumerate}

\bibliographystyle{acmart-primary/ACM-Reference-Format}
\bibliography{main}

\end{document}